%% file: main.tex
\documentclass[11pt]{article}

\pdfoutput=1

\input{flowmacs}

\newcommand{\setR}{\mathbb{R}}

\begin{document}

\title{\bf Near-Optimal Distributed Maximum Flow}

\author{
Mohsen Ghaffari\\
Massachusetts Institute of Technology\\
Cambridge MA, USA\\
{\small \texttt{ghaffari@csail.mit.edu}}
\and
Andreas Karrenbauer\\
MPI for Informatics\\
Saarbr\"ucken, Germany\\
{\small \texttt{karrenba@mpi-inf.mpg.de}}
\and
Fabian Kuhn\\ 
Univerity of Freiburg\\
Freiburg, Germany\\
{\small \texttt{kuhn@cs.uni-freiburg.de}}
\and
Christoph Lenzen\\ 
MPI for Informatics\\
Saarbr\"ucken, Germany\\
{\small \texttt{clenzen@mpi-inf.mpg.de}}
\and
Boaz Patt-Shamir\\
Tel Aviv University\\
Tel Aviv, Israel\\
{\small \texttt{boazps@post.tau.ac.il}}
}
\date{}

\maketitle

\begin{abstract}
We present a near-optimal distributed algorithm for $(1+o(1))$-approximation of
single-commodity maximum flow in undirected weighted networks that runs in $(D+
\sqrt{n})\cdot n^{o(1)}$ communication rounds in the \Congest model. 
Here, $n$ and $D$ denote the number of nodes and the network diameter,
respectively. This is the first improvement over the trivial bound of $O(n^2)$,
and it nearly matches the $\tilde{\Omega}(D+ \sqrt{n})$ round complexity lower 
bound. 

The development of the algorithm contains two results of
independent interest:
\begin{compactenum}[(i)]
  \item A $(D+\sqrt{n})\cdot n^{o(1)}$-round distributed construction 
  of a
  spanning tree of average stretch $n^{o(1)}$.
  \item A $(D+\sqrt{n})\cdot n^{o(1)}$-round distributed construction 
  of an
  $n^{o(1)}$-congestion approximator consisting of the cuts induced by $O(\log
  n)$ virtual trees. The distributed representation of the cut approximator
  allows for evaluation in $(D+\sqrt{n})\cdot n^{o(1)}$ rounds.
\end{compactenum}
All our algorithms make use of randomization and succeed with high probability.
\end{abstract}

\listoftodos

\input{intro}

\section{Overview of the Centralized Framework}\label{sec:overview}

Sherman's approach~\cite{Sherman13} is based on \emph{gradient
descent} (see, e.g.,~\cite{nesterov2004introductory-book}) for \emph{congestion minimization} with a clever 
dualization of the flow conservation constraints. The flow problem is
re-formulated as a \emph{demand vector} $\vb \in \Reals^n$ such that $\sum_{i\in
V} b_i=0$. In the case of the $s$-$t$ flow problem, we have a positive $b_s$
and negative $b_t$ with the same absolute value and the demand is zero
everywhere else. The objective is to find a flow $\vf^*$ that meets the given
demand vector, i.e., the total excess flow in node $i$ is equal to $b_i$, and
minimizes the maximum \emph{edge congestion}, which is the ratio of the flow
over an edge to its capacity. Formally:
\begin{equation}\label{congestion minimization}
\text{minimize } \maxnorm{C^{-1}\vf} \text{ subject to } B\vf=\vb\,,
\end{equation}
where $C=(C_{ee'})_{e,e'\in E}$ is an $m\times m$ diagonal matrix with
\begin{equation*}
C_{ee'}=\left\{\begin{matrix}
\Cp(e) & \mbox{if }e=e'\\
0 & \mbox{else\,,}
\end{matrix}\right.
\end{equation*}
and $B=(B_{ve})_{v\in V,e\in E}$ is an $n\times m$ matrix with
\begin{equation*}
B_{ve}=\left\{\begin{matrix}
1 & \mbox{if }e=(u,v)\mbox{ for some }u\in V\\
-1 & \mbox{if }e=(v,u)\mbox{ for some }u\in V\\
0 & \mbox{else\,.}
\end{matrix}\right.
\end{equation*}
Note that given a general (i.e., unconstrained) flow vector $\vf\in\Reals^m$,
$(B\vf)_v$ is exactly the excess flow at node $v$. Hence, by the max-flow
min-cut theorem, if we can solve problem~\eqref{congestion minimization}, a
simple binary search will find an approximate max flow.

Instead of directly solving this constrained system, Sherman allows for general
flows and adds a penalty term for any violation of flow constraints, i.e.,
$$\text{minimize } \maxnorm{C^{-1}\vf} + 2\alpha \maxnorm{R(\vb-B\vf)}\,,$$ where
$\alpha\geq 1$ and the matrix $R$ are chosen so that the optimum of this
unconstrained optimization problem does not violate the flow constraints. As we
are interested in an approximate max flow, we can compute an approximate
solution and argue that the violation of the flow constraints will be small,
too. Then one simply re-routes the remaining flow in a trivial manner, e.g.\ on
a spanning tree, to obtain a near-optimal solution. Finally, to ensure that the
objective function is differentiable (i.e., a gradient descent is actually
possible), $\maxnorm{\cdot}$ is replaced by the so-called soft-max. 

\paragraph{The Congestion Approximator $R$.}
The \emph{congestion} of an edge $e$ (for a given flow $\vf$) is 
defined as the ratio
$|f_e|/\Cp(e)$. When referring to the congestion of a cut in a given 
flow, we mean 
the ratio between the net flow crossing the cut to the total capacity 
of the cut.
Suppose for a moment that $\alpha=1$
and $R$ contains one row for each cut of the graph, chosen such that 
each entry of the vector $R B \vf$ equals the congestion of the
corresponding cut. In particular, $R$ would correctly reproduce the 
congestion of min 
cuts (which give rise to maximal congestion). 
Moreover, the vector $R\vb$ describes the inevitable congestion of the 
cuts for any feasible flow.
Thus, the components of $R(\vb-B\vf)$ are the residual congestions to 
be dealt with to make $\vf$ feasible 
(neglecting possible cancellations).
The max-flow min-cut theorem and the factor of $2$ in the second term 
of the
objective function imply that it always improves the value of the
objective function to route the demands arising from a violation of flow
constraints optimally.
Moreover, the gradient descent concentrates on the most congested edges 
and those that are
contained in cuts with the top residual congestion. In particular, 
flow is pushed over the edges into the cut with the highest residual congestion 
to satisfy its demand until other cuts become more important in the second part of the objective. 
The first part of the objective impedes flow on edges the more they are 
congested 
(on an absolute scale and relative to others).
Thus,
approximately minimizing the objective function is equivalent to simultaneously
approximating the minimum congestion and having small violation of flow
constraints; solving up to polynomially small error and naively resolving the
remaining violations then yields sufficiently accurate results.

Unfortunately, trying to make $R$ capture congestion \emph{exactly} is far too
inefficient. Instead, one uses an \emph{$\alpha$-congestion approximator}, that
is a matrix $R$ such that for any demand vector $\vb$, it holds that
\begin{equation*}
\maxnorm{R\vb}\leq \operatorname{opt}(\vb) \leq \alpha\maxnorm{R\vb},
\end{equation*}
where $\operatorname{opt}(\vb)$ is the maximum congestion caused on any cut by
optimally routing $\vb$. Since the second term in the objective function is
scaled up by factor $\alpha$, we are still guaranteed that optimally routing any
excess demands improves the objective function. However, this implies that the
second term of the objective function may dominate its gradient 
and thus emphasis is shifted rather to feasiblity than optimality.
Sherman proves that this slows down the
gradient descent by at most a factor of $\alpha^2$, i.e., if $\alpha \in
n^{o(1)}$, so is the number of iterations of the gradient descent algorithm that
need to be performed.

\paragraph{Congestion Approximators: R\"acke's Construction.}
For any spanning tree $T$ of $G$, deleting an edge partitions the nodes into two
connected components and thus induces an (edge) cut of $G$. Note that on $T$,
this cut contains only the single deleted edge, and in terms of congestion any
cut of $T$ is dominated by such an edge-induced cut: For any cut, the maximum
congestion of an edge is at least the average congestion of the cut, and in $T$,
there is a cut containing only this edge.

These basic properties motivate the question of how well the cut 
structure of an
arbitrary graph can be approximated by trees. Intuitively, the goal is
to find a tree $T$ (not necessarily a subgraph) spanning all nodes
with edge weights such that routing \emph{any} demand vector in $G$ and 
in $T$
results in roughly the same maximal congestion. Because routing flows on
trees is trivial, such a tree $T$ would give rise to an efficient
congestion approximator $R$:
$R$ would consist of one row for each cut induced by an edge $(u,v)$ of $T$ with
capacity $C$, where the matrix entry corresponding to node $w$ is $1/C$ if $w$
is on $u$'s ``side'' of the cut and $0$ otherwise; multiplying a demand vector
with the row then yields the flow that needs to pass through $(u,v)$ divided by
the capacity of the cut.

In a surprising result~\cite{raecke08}, R\"acke showed that, using
multiplicative weight updates (see e.g.~\cite{arora2012multiplicative,
plotkin1995fast, young2001sequential}) one can construct a distribution of
$\tilde{O}(m)$ trees so that (i) in each tree of the distribution, each cut has
at least the same capacity as in $G$ and (ii) given any cut of $G$ of
total capacity $C$, sampling from the distribution results in a tree $T$ where
this cut has \emph{expected} capacity $O(\alpha C)$; here $\alpha$ is the
approximation ratio of a \emph{low average stretch spanning tree} algorithm
R\"acke's construction uses as subroutine. Note that this bound on the
expectation implies that for any cut of capacity $C$, there must be a tree in
the distribution for which the cut has capacity $O(\alpha C)$. Hence, the cuts
given by \emph{all} trees in the distribution give rise to an
$O(\alpha)$-congestion approximator $R$ with $\tilde{O}(m n)$ rows.

\paragraph{Low Average Stretch Spanning Trees.}
In order to perform R\"acke's construction, one requires an efficient algorithm
for computing low average stretch spanning trees. More precisely, given a
graph $G=(V,E,\ell)$ with polynomially bounded lengths $\ell:E\to \Nats$, the
goal is to construct a spanning tree $T$ of $G$ so that
\begin{equation*}
\sum_{\{u,v\}\in E}d_T(u,v)\leq \alpha \sum_{\{u,v\}\in E}\ell(\{u,v\})\,,
\end{equation*}
where $d_T(u,v)$ is the sum of the lengths of the unique path from $u$ to $v$ in $T$
and $\alpha$ is the \emph{stretch} factor.

Sherman's algorithm builds on a sophisticated low average stretch spanning tree
algorithm that achieves $\alpha \in O(\log n \log^2 \log n)$
within $\tilde{O}(m)$ centralized steps~\cite{abraham2008nearly}. We use
a simpler approach providing $\alpha \in 2^{O(\sqrt{\log n \log \log
n})}$~\cite{alon1995graph} that has been shown to parallelize well, i.e., has an
efficient implementation in the PRAM model~\cite{blelloch14}.

\paragraph{Congestion Approximators: Madry's Construction.}
R\"acke's construction has the drawback that one needs to sequentially compute a
linear number of trees, which is prohibitively expensive from our point 
of view. Madry generalized R\"acke's approach to a
construction that results in a distribution over $\tilde{O}(m/j)$ so-called
\emph{$j$-trees}~\cite{Madry10}, where $j$ is a parameter. A $j$-tree 
consists of a \emph{forest} of $j$ connected components (trees) and a 
\emph{core graph}, which is an arbitrary 
connected graph with $j$ nodes: one from each tree (see 
\figureref{fig:j-tree}).
\setlength{\columnsep}{15pt}
\begin{wrapfigure}{r}{0.30\textwidth}
	\vspace{-15pt}
	\begin{center}
  	\includegraphics[width=40 mm]{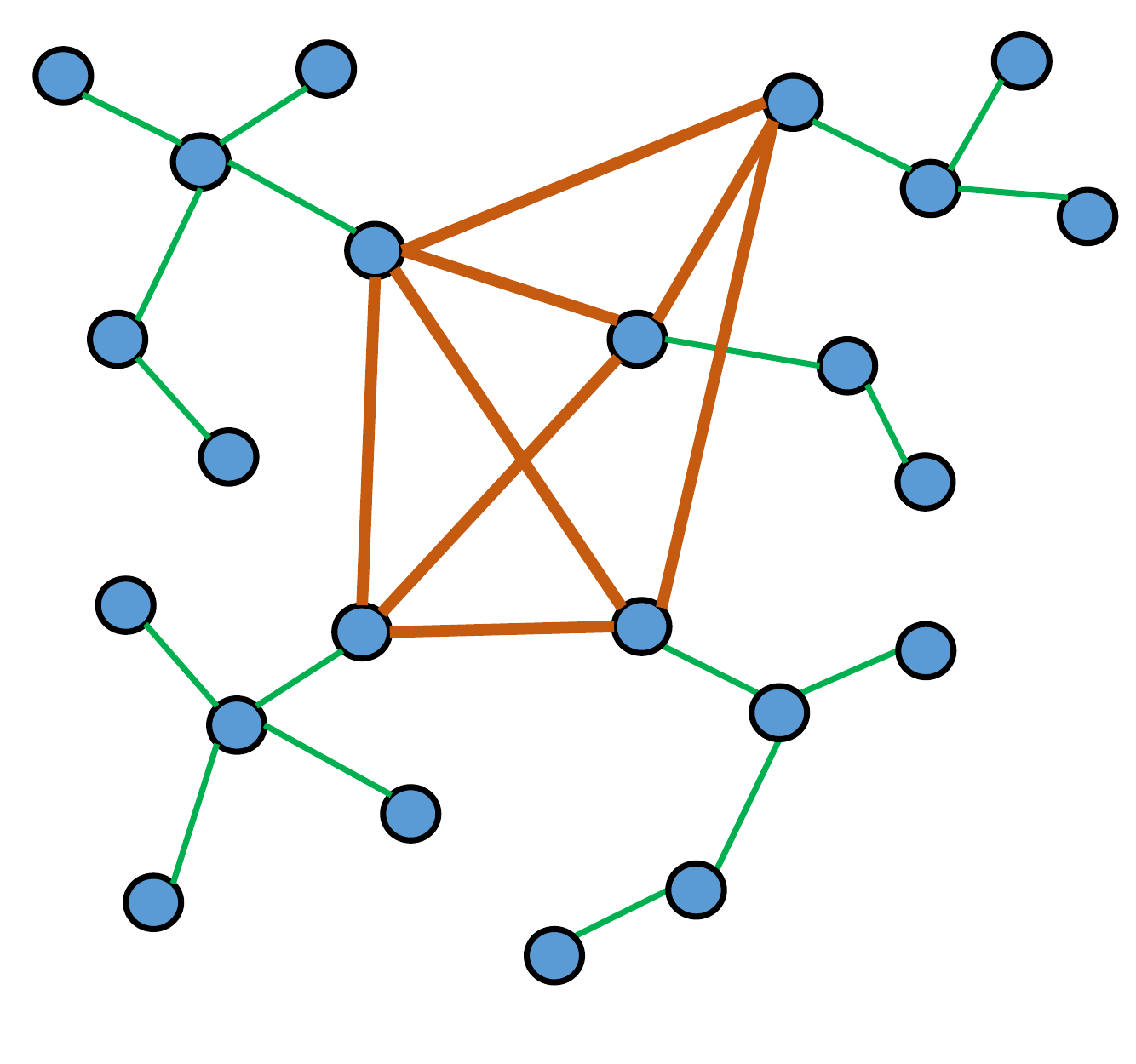}
  \end{center}
	\vspace{-25pt}
	\caption{\small A $j$-tree for $j=5$. The core links are depicted in brown.}
	\vspace{-10pt}
	\label{fig:j-tree}
\end{wrapfigure}

The properties of the distribution are the same as for R\"acke's: sampling from
the distribution preserves cut capacities up to an expected $O(\alpha)$-factor,
where $\alpha$ is the stretch of the utilized spanning tree algorithm. Likewise,
using all (dominant) cuts of all $j$-trees in the distribution to construct $R$
yields an $O(\alpha)$-congestion approximator. Note that any cut in a
$j$-tree is \emph{dominated} by either a cut induced by an edge of the 
forest, or by a cut of the core, in the following sense:
Consider any demand vector and any ``mixed'' cut. If there is an edge in the
forest crossing the cut that has at least the same congestion as the whole cut,
then the cut induced by the forest edge dominates the mixed cut. Otherwise, we
can remove all forest edges from the mixed cut without reducing its congestion.
As routing demands in the forest part of the graph is trivial, Madry's
construction can be seen as an efficient reduction of the problem size.

\paragraph{Congestion Approximators: Combining Cut Sparsifiers with Madry's
Construction.}
Using $j$-trees, Sherman derives a suitable congestion approxmiator, 
i.e.,
one with $\alpha\in n^{o(1)}$ that can be constructed and evaluated in
$\tilde{O}(m+n^{1+o(1)})$ rounds, as follows. First, a \emph{cut sparsifier} is
applied to $G$. A $(1+\eps)$-sparsifier computes a subgraph of
$G$ with modified edge weights so that the capacities of all cuts are 
preserved
up to factor $1+\eps$. It is known how to compute a $(1+o(1))$-sparsifier
with $\tilde{O}(n)$ edges in $\tilde{O}(m)$ steps using
randomization~\cite{BK15}. As the goal is merely to compute a
congestion approximator
with $\alpha\in n^{o(1)}$, the multiplicative $1+o(1)$ approximation error is
negligible. Hence, this essentially breaks the problem of computing a congestion
approximator down to the same problem on sparse graphs.

Next, Sherman applies Madry's construction with $j=n/\beta$, where
$\beta=2^{\sqrt{\log n}}$. This yields a distribution of $\tilde{O}(\beta)$ many
$n/\beta$-trees. The issue is now that the cores are arbitrary graphs, implying
that it may be difficult to evaluate congestion for cuts in the cores. However,
the number of nodes in the core is $n'=n/\beta$. Thus, recursion does the trick:
apply the cut sparsifier to the core, use Madry's construction on the resulting
graph (with $j'=n'/\beta=n/\beta^2$), rinse and repeat. In total, there are
$\log_{\beta} n = \sqrt{\log n}$ levels of recursion until the core becomes
trivial, i.e., we arrive at a tree. For each level of recursion, the
approximation ratio deteriorates by a multiplicative $\alpha\in \polylog n$,
where $\alpha$ is the stretch factor of the low-stretch spanning tree algorithm,
and a multiplicative $1+o(1)$, for applying the cut sparsifier.
This yields an $\alpha'$-congestion approximator with
\begin{equation*}
\alpha'\in ((1+o(1))\alpha)^{\sqrt{\log n}}\subset 2^{O(\sqrt{\log n}\log \log
n)}\subset n^{o(1)}\,.
\end{equation*}
While the total number of constructed trees is still
$\tilde{O}(\beta^{\log_{\beta} n})=\tilde{O}(n)$, the number of nodes in a graph
(i.e., a core from the previous level) on the $i^{th}$ level of recursion is
only $n/\beta^{i-1}$. The cut sparsifier ensures that the number of edges in
this graph is reduced to $\tilde{O}(n/\beta^{i-1})$ \emph{before} recursing.
Since the number of edges in the core is (trivially) bounded by the number of
edges of the graph in Madry's construction, the total number of 
sequential computation
steps for computing the distribution is thus bounded by
\begin{equation*}
\tilde{O}(m)+\sum_{i=1}^{\log_{\beta} n} \tilde{O}(\beta^i \cdot n/\beta^{i-1})
\subset \tilde{O}(m+n^{1+o(1)})\,.
\end{equation*}

\paragraph{Step Complexity of the Flow Algorithm.}
The above recursive structure can also be exploited to \emph{evaluate} the
$\alpha'$-congestion approximator Sherman uses in $n^{1+o(1)}$ steps. As
mentioned earlier, the cuts of a $j$-tree are dominated by those induced by
edges of the forest and those which are crossed by core edges only
(cf.~\figureref{fig:j-tree}). In the forest component, routing demands is
unique, takes linear time in the number of nodes (simply start at the leaves),
and results in a modified demand vector at the core on which is recursed.

Sherman proves that his algorithm obtains a $(1+\eps)$-approximate flow
in $O(\eps^{-3}\alpha^2 \log^2 n)$ gradient descent steps, 
provided
$R$ is an $\alpha$-congestion approximator.\footnote{Sherman points out that
using Nesterov's accelerated gradient descent method~\cite{nesterov2005smooth},
this can be improved to $O(\eps^{-2}\alpha \log^2 n)$. For both his and
our results, this difference is insubstantial, as $\alpha\in
n^{o(1)}\Leftrightarrow \alpha^2 \in n^{o(1)}$.} It is straightforward to see
(cf.~\sectionref{subsec:almostRoute}) that each of these steps requires $O(m)$
computational steps besides doing two matrix-vector multiplications with $R$ and
$R^{\top}$, respectively. Using the above observation and plugging in the time
to construct the (implicit) representation of $R$, one arrives at a total step
complexity of $\tilde{O}(m n^{o(1)})$.

\todo{Would we like to explain the quasi-centralized near-linear time algorithm
first?}

\section{Distributed Algorithm: Contribution and Key 
Ideas}\label{sec:contribution}

\paragraph{The Distributed Toolchain.}
For a distributed implementation of Sherman's approach, many subproblems need
to be solved (sufficiently fast) in the \Congest model. We summarize them in the
following list, where stars indicate that these components are readily available
from prior work.
\begin{compactenum}
\item [*] Decomposing trees into $O(\sqrt{n})$ components of strong diameter
$O(\sqrt{n})$, within $\tilde{O}(\sqrt{n}+D)$ rounds. This can, e.g., be done by
techniques pioneered by Kutten and Peleg for the purpose of minimum-weight
spanning tree construction~\cite{Kutten-Peleg}.
\item [*] Constructing cut sparsifiers. Koutis~\cite{Koutis14} provides a
solution that completes in $\polylog n$ rounds of the \Congest model. In
\sectionref{sec:sparsifiers}, we give a simulation result for use in the
recursive construction.
\item Constructing low average stretch spanning trees on multigraphs
(\sectionref{sec:spanning}).
\item Applying Madry's construction in the \Congest model, even when recursing
in the context of Sherman's framework (\sectionref{sec:jtrees}).
\item Sampling from the recursively constructed distribution
(\sectionref{sec:jtrees}).
\item Avoiding the use of the entire distribution for constructing the
congestion approximator (see below).
\item Performing a gradient descent step. This involves, e.g., matrix-vector
multiplications with $R$, $R^{\top}$ and $C^{-1}$, evaluation of the soft-max,
etc.\ (\sectionref{sec:gradient}).
\end{compactenum}

\bigskip
\noindent We next present some additional description for the items 1 to 5 in
the above list.

\paragraph{1. Low Average Stretch Spanning Trees.}
In \sectionref{sec:spanning}, we prove the following theorem.
\begin{theorem}\label{theorem:spanning}
Suppose $H$ is a multigraph obtained from $G$ by assigning arbitrary
edge lengths in $[2^{n^{o(1)}}]$ to the edges of $G$ (known to incident nodes) and
performing an arbitrary sequence of contractions. Then we can compute a
spanning tree of $H$ of expected stretch $2^{O(\sqrt{\log n \log \log n})}$
within $(\sqrt{n}+D)n^{o(1)}$ rounds.
\end{theorem}
To obtain this theorem, we translate a PRAM algorithm by Blelloch et
al.~\cite{blelloch14} to the \Congest model. The main issue when transitioning
from the PRAM to the \Congest model is that in the PRAM model, information about
distant parts of the graph may be readily accessed. In the \Congest model, we
handle this by pipelining long-distance communication
over a global breadth-first-search (BFS) tree of $G$; communication 
over $O(\sqrt{n})$
hops is handled using the edges that have already been selected for inclusion
into the spanning tree and spanning trees of the contracted regions of $G$.

\paragraph{2. Implementing Madry's Scheme.}
This is  technically the most challenging part. Also here, we have to 
overcome
the difficulty of potentially needing to communicate a large amount of
information over many hops; doing this naively results in too much contention
and thus slow algorithms. We approach this by modifying Madry's construction so
that:
\begin{compactitem}
\item Instead of ``aggregating'' edges so that the core becomes a graph, we
admit a multigraph as core.
\item We do not explicitly construct the core. Instead, we simulate both the
sparsifier and the low average stretch spanning tree algorithm using the
abstraction of \emph{cluster graphs} (see \sectionref{sec:cluster}).
\item In doing so, we maintain that every core edge is also a graph edge. This
enables to handle all communication over this edge by using the corresponding
graph edge.
\item In this context, clusters are the forest components rooted at core nodes.
We will maintain that forest components have depth $\tilde{O}(\sqrt{n})$. While
this is not strictly necessary, it simplifies the description of the
corresponding distributed algorithms, as the communication within each cluster
can then be performed via its (previously constructed) spanning tree.
\item The cluster hierarchy that is established during the construction allows
for a straightforward recursive evaluation of the corresponding congestion
approximator.
\end{compactitem}
\sectionref{sec:jtrees} gives the details of the construction.

\paragraph{3. Sampling from the Distribution.}
This is now straightforward, because for each sample, on each level of the
recursion we need to construct only $n^{o(1)}$ different $j$-trees for some $j$.
This is also discussed in \sectionref{sec:jtrees}, in which the formal version
of the following theorem is proved.
\begin{theorem}[Informal]\label{theorem:sample}
Within $\tilde{O}((\sqrt{n}+D)\beta)$ rounds of the \Congest model, we can
sample a virtual tree from the distribution used in Sherman's framework, where
$\tilde{O}(\beta)$ is the number of $j$-trees in the distribution constructed
when recursing on a core. The distributed representation allows to evaluate the
dominant cuts of the tree when using it in a congestion approximator within
$\tilde{O}(\sqrt{n}+D)$ rounds.
\end{theorem}

\paragraph{4. Avoiding the use of the entire distribution for constructing the
congestion approximator.}
While Sherman can afford to use all trees in the (recursively constructed)
distribution, the above theorem is not strong enough to allow for fast
evaluation of all $\tilde{\Theta}(n)$ trees. As Madry points out~\cite{Madry10},
it suffices to sample and use $O(\log n)$ $j$-trees from the distribution he
constructs to speed up any $\beta$-approximation algorithm for an ``undirected
cut-based minimization problem'', at the expense of an increased approximation
ratio of $2\alpha \beta$, where $\alpha$ is the approximation ratio of the
congestion approximator corresponding to the distribution of $j$-trees. The
reasoning is as follows:
\begin{compactitem}
\item The number of cuts that need to be considered for such a problem is
polynomially bounded.
\item The expected approximation ratio for any fixed cut when sampling from the
distribution is $\alpha$. By Markov's bound, with probability at least $1/2$ it
is at most $2\alpha$.
\item For $O(\log n)$ samples, the union bound shows that w.h.p.\ \emph{all}
relevant cuts are $2\alpha$-approximated.
\item Applying a $\beta$-approximation algorithm relying on the samples
only, which can be evaluated much faster, results in a
$2\alpha\beta$-approximation w.h.p.
\end{compactitem}
Recall that the problem of approximating a max flow was translated to
minimizing congestion for demands $F$ and $-F$ at $s$ and $t$ and performing
binary search over $F$. The max-flow min-cut theorem implies the respective
congestion to be the function of a single cut, which can be used to verify that
the problem falls under Madry's definition.

Unfortunately, applying the sampling strategy as indicated by Madry is
infeasible in Sherman's framework. As the goal is a
$(1+\eps)$-approximation, applying it to the above problem directly will
yield a too inaccurate approximation. Alternatively, we can apply it in the
construction of a congestion approximator. However, a congestion approximator
must return a good approximation for \emph{any} demand vector. There are
exponentially many such vectors even if we restrict $\vb\in \{-1,0,1\}^V$, and
we are not aware of any result showing that the number of min-cuts corresponding
to the respective optimal flows is polynomially bounded.

We resolve this issue with the following simple, but essential insight, at the
expense of squaring the approximation ratio of the resulting congestion
approximator.
\begin{lemma}\label{lemma:approximator}
Suppose we are given a distribution of $\poly n$ trees so that given any cut of
$G$ of capacity $C$, sampling from the distribution results in a tree whose
corresponding cut has at least capacity $C$ and at most capacity $\alpha C$ in
expectation. Then sampling $O(\log n)$ such trees and constructing a congestion
approximator from their single-edge induced cuts results in a
$2\alpha^2$-congestion approximator of $G$ w.h.p.
\end{lemma}
\begin{proof}
Recall that a cut approximator estimates the maximum congestion when optimally
routing an arbitrary demand. Consider any demand vector and denote by $C$ the
capacity of the corresponding cut that is most congested when routing the
demand. As sampling from the distribution yields approximation factor $\alpha$
in expectation, there must be \emph{some} tree $T$ in the distribution whose
corresponding cut has capacity at most $\alpha C$. However, this means that when
routing the demand via $T$, there is some edge in $T$ that experiences at least
$1/\alpha$ times the maximum congestion when routing the demand optimally in
$G$. As the capacity of the edge is at least that of the corresponding cut in
$G$, it follows that the corresponding cut of $G$ has congestion at least
$1/\alpha$ of that of the min-cut when routing the demand.

As there are $\poly n$ trees, each of which has $n-1$ edges, this shows that for
any demand vector there is one of polynomially many cuts of $G$ that experience
at least $1/\alpha$ times the maximum congestion when optimally routing the
demand vector. By Markov's bound and the union bound, w.h.p.\ the congestion on
each of these cuts will be approximated up to another factor of $2\alpha$ when
using $O(\log n)$ samples.
\end{proof}

\todo{Do we want to elaborate on how to obtain faster centralized algorithms
using this trick, too? While not a new result, it might be interesting from the
angle of providing a different technique.}

\paragraph{5. Performing a gradient descent step.}
Most of the high-level operations required for executing a gradient descent
algorithm are straightforward to implement using direct communication between
neighbors or broadcast and convergecast operations on a BFS tree. The most
involved part is multiplying the (implicitly constructed) congestion
approximator $R$ with an arbitrary demand vector $\vb$, and multiplying the
transposed of the approximator matrix, $R^{\top}$, with a given vector that
specifies a \emph{cost} for each edge of the trees.

Multiplying by $R$ is done by exploiting that routing on trees is trivial and
using standard techniques: during the construction, we already decomposed each
tree into $O(\sqrt{n})$ components of strong diameter $O(\sqrt{n})$, which
can be used to solve partially by contracting components, making the resulting
tree of $O(\sqrt{n})$ nodes globally known, then determine modified demand
vectors for the components out of the now locally computable partial solution,
and finally resolve these remaining demands within each component.
Multiplication with $R^{\top}$ is implemented using similar ideas. We refer to
\sectionref{sec:gradient} for a detailed discussion of these procedures.
Plugging the building blocks outlined in this section into this machinery, we
obtain our main result \theoremref{theorem:main_informal}.

\section{Outline of Distributed Congestion Approximator
Construction}\label{sec:approximator}

In this section, we outline how to adapt Madry's construction to its
recursive application in the distributed setting. In \sectionref{sec:jtrees},
we formally prove that we achieve the same guarantees as Madry's
distribution~\cite{Madry10} in each recursive step and that our distributed
implementation is fast. Here, we focus on presenting the main ideas of the
required modifications to Madry's scheme and its distributed implementation; to
this end, it suffices to consider the construction of a single step of the
recursion.
\paragraph{Centralized Algorithm.}
As a starting point, let us summarize the main steps of one iteration of the
centralized construction. We state a slightly simplified variant of Madry's
construction, which offers the same worst-case performance and is a better
starting point for what follows. From the previous step of constructing the
distribution, an edge length function $\ell_e$ is known (in the distributed
setting, this knowledge will be local). Given $j\leq n-1$, the following
construction yields a $\Theta(j)$-tree.
\begin{compactenum}
\item Compute a spanning tree $\cT$ of $G$ of stretch $\alpha$.
\item For each edge $e=\{v,w\}\in E$ of the graph $G$, route $\Cp(e)$ units of
a commodity $\operatorname{com}_e$ from $v$ to $w$ on (the unique path from $v$
to $w$ in) $\cT$.\footnote{The difference to a single commodity is simply
that flows in opposing directions do not cancel out. This means that \emph{any}
given feasible (i.e., congestion-$1$) flow in $G$ can be routed on $T$ with at
most the congestion of this multi-commodity flow.} Denote by $\vf$ the vector
of the sum of absolute flows passing through the edges of $\cT$. Recall that
$\max_{e\in E}\{\Cp(e)\}\in \poly n$ and thus $\maxnorm{\vf}\in \poly(n)$.
\item For $e\in \cT$, define the \emph{relative load} of $e$ as
$\rload(e):=|f_e|/\Cp(e)\in \poly n$. We decompose the edge set of $\cT$ into 
$O(\log n)$ subsets $\cF_i$, $i\in \{1,\ldots,\lceil \log
(\maxnorm{\vf}+1)\rceil\}$, where $e\in \cT$ is in $\cF_i$ if $\rload(e)\in
(R/2^i,R/2^{i-1}]$ for $R:=\max_{e\in \cT}\{\rload(e)\}$. As $\cT$ has $n-1\geq
j$ edges, there must be some $\cF_i$ with $\Omega(j/\log n)$ edges; let $i_0$ be
minimal with this property. Define $\cF:=\{e\in \cT\,|\,\rload(e)>2^{i_0-1}\}$.
Note that $|\cF|\leq j$.
\item $\cT\setminus \cF$ is a spanning forest of at most $j+1$ components.
Define $H$ as the graph on node set $V$ whose edge set is the union of
$\cT\setminus \cF$ and all edges of $G$ between different components of
$(V,\cT\setminus \cF)$.
\item For components $C$ and $C'$ of $(V,\cT\setminus \cF)$, pick arbitrary
$v\in C$ and $w\in C'$ and denote by $p(C,C')\in C$ the last node from $C$ on
the $v$-$w$ path in $\cT$; note that $p(C,C')$ does not depend on the choice of
$v$ and $w$. Denote by $P$ the set of such \emph{portals}. Replace all edges
between different components $C,C'$ of $(V,\cT\setminus \cF)$ by parallel edges
$\{p(C,C'),p(C',C)\}$ (of the same weight).
\item In the resulting multigraph, iteratively delete nodes from $V\setminus
P$ of degree $1$ until no such node remains. Note that the leaves of the induced
subtree of $\cT$ must be in $P$, showing that the number of remaining nodes in
$V\setminus P$ of degree larger than $2$ is bounded by $|P|-1< 2j$. Add all
such nodes to $P$.
\item For each path with endpoints in $P$ and no inner nodes in $P$, delete an
edge of minimum capacity and replace it by an edge of the same capacity between
its endpoints.
\item Re-add the nodes and edges of $\cT\setminus \cF$ that have been deleted in
Step~$6$.
\item For any $p,q\in P$, merge all parallel edges $\{p,q\}$ into a single one
whose capacity is the sum of the individual capacities. The result is a
$j'$-tree for $j'=|P|<4j$.
\end{compactenum}
In his paper, Madry provides a scheme for updating the edge lengths between
iterations so that this construction results in a distribution on
$\tilde{O}(m/j)$ $\Theta(j)$-trees that approximate cuts up to an expected
$O(\alpha)$-factor, where $\alpha$ is the stretch of the spanning tree
construction. Updating the edge length function poses no challenges, so we will
focus on the distributed implementation of the above steps in this section.

\paragraph{Differences to the Centralized Algorithm.}
Before we come to the distributed algorithm, let us first discuss a few changes
we make to the algorithm in centralized terms. These do not affect the reasoning
underlying the scheme, but greatly simplify its distributed implementation.
\begin{compactitem}
\item We will omit the last step of the algorithm and instead operate on cores
that are multigraphs. This changes the computed distribution, as we formally use
a different graph as input to the recursion. However, R\"acke's arguments (and
Madry's generalization) work equally well on multigraphs, as one can see by
replacing each edge of the multigraph by a path of length $2$, where both edges
have the same capacity as the original edge. This recovers a graph of $2m$ edges
from a multigraph of $m$ edges without affecting the cut structure, and the
resulting trees can be interpreted as trees on the multigraph by contraction of
the previously expanded edges. Similarly, both the low average stretch spanning
tree construction and the cut sparsifier work on multigraphs without
modification.
\item After computing the spanning tree, we will immediately delete a subset of
$\tilde{O}(\sqrt{n})$ edges to ensure that the new clusters will have low-depth
spanning trees. The deleted edges are replaced by all edges of $G$ crossing the
corresponding cuts and will end up in the core. The same procedure is, in fact,
applied to all edges selected into $\cF$ in Step 3 of the centralized routine;
Madry's arguments show that removing \emph{any} subset of edges of $\cT$ and
replacing it this way can only improve the quality of cut approximation. The
main point of his analysis is that choosing $\cF$ in the way he does guarantees
that, in terms of constructing the final distribution of $j$-trees, progress
proportional to the number of edges in $\cR_{i_0}$ is made. We will apply the
construction to cores of size $n'\gg \tilde{O}(\sqrt{n})$, which implies that
removing the additional edges has asymptotically no effect on the progress
guarantee.
\item In the counterpart to Step~6 in Madry's routine, also nodes from $P$ may
be removed if their degree becomes $1$. Also here, there is no asymptotic
difference in the worst-case performance of our routine from Madry's.
\end{compactitem}

To simplify the presentation, in this section we will assume that all trees
involved in the construction have depth $\tilde{O}(\sqrt{n})$. This means that
we can omit the deletion of $\tilde{O}(\sqrt{n})$ additional edges and further
related technicalities. The general case is handled by standard techniques for
decomposing trees into $O(\sqrt{n})$ components of depth $\tilde{O}(\sqrt{n})$
and relying on a BFS tree to communicate ``summaries'' of the components to all
nodes in the graph within $\tilde{O}(\sqrt{n}+D)$ rounds (full details are given
in \sectionref{sec:jtrees}). This approach was first used for MST
construction~\cite{Kutten-Peleg}; we use a simpler randomized variant
(cf.~\lemmaref{lemma:smallstrongdiameter}).

\paragraph{Cluster Graphs.}
Recall that we will recursively call (a variant of) the above
centralized procedure on the core. We need to simulate the algorithm on the core
by communicating on $G$. To this end, we will use \emph{cluster graphs}
(see~\sectionref{sec:cluster}), in which $G$ is decomposed into components that
play the role of core nodes. We will maintain the following invariants during
the recursive construction:

\begin{compactenum}
\item There is a one-to-one correspondence between core nodes and
\emph{clusters}.
\item Each cluster $c$ has a rooted spanning tree of depth
$\tilde{O}(\sqrt{n})$.
\item No other edges exist inside clusters. Contracting clusters yields the
multigraph resulting from the above construction without Step~$9$.
From now on, we will refer to this multigraph as the core.
\item All edges in the (non-contracted) graph are also edges of $G$, and their
endpoints know their lengths from the previous.
\end{compactenum}


\paragraph{Overview of the Distributed Routine.}
We follow the same strategy as the centralized algorithm, with the modifications
discussed above. This implies that the core edges for the next recursive call
will simply be the graph edges between the newly constructed clusters. The
following sketches the main steps of the distributed implementation of our
overall approach.

\begin{figure}[t]
	\centering
	\includegraphics[width=0.5\textwidth]{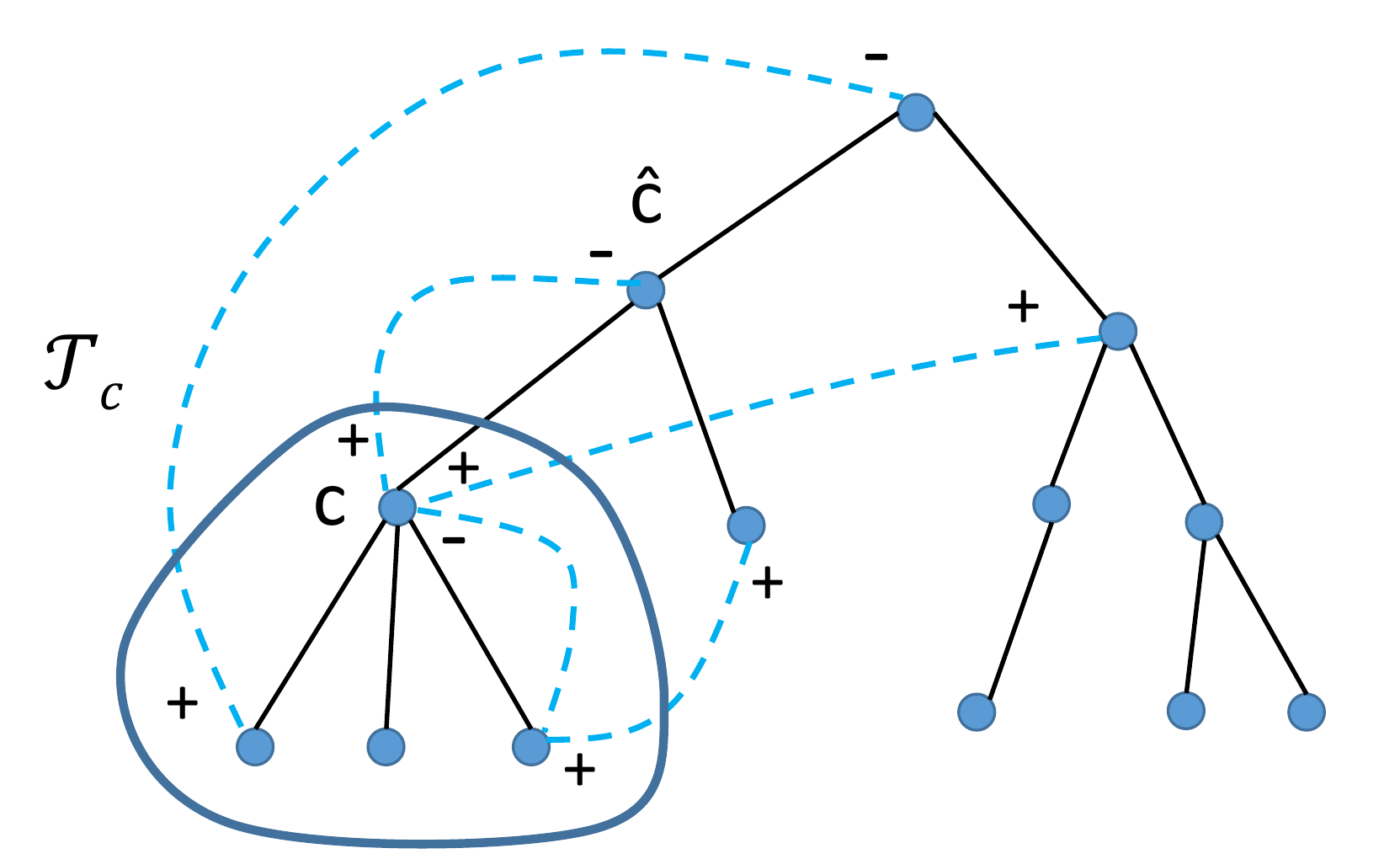}
	\caption{\small Illustration of the underlying idea of the 
	aggregation scheme for the
		cut capacities. The cut corresponding to edge $(c,\hat{c})$ of 
		the tree has a
		total capacity given by all graph edges leaving the subtree 
		$\cT_c$. By
		labeling the endpoint of graph edge by ``+'' if it leaves the 
		subtree and by
		``-'' if it connects to a descendant, the cut capacity is thus 
		the sum of all
		capacities of edges labeled ``+'' minus all those of edges 
		labeled ``-'' within
		$\cT_c$.}
	\label{fig:plusMinus}
\end{figure}

\begin{compactenum}
\item Compute a spanning tree $\cT$ of stretch $\alpha$ of the core. This is
done by the spanning tree algorithm of \theoremref{theorem:spanning}, which can
operate on the cluster graph.
\item For each edge $e\in \cT$, determine its absolute flow $|f_e|$ (and
thus $\rload(e)=|f_e|/\Cp(e)$) as follows (cf.~\figureref{fig:plusMinus}).
\begin{compactenum}
\item [(*)] For each cluster $c$, consider the cut induced by the edge to
its parent. For each ``side'' of the cut, we want to determine the total
capacity of all edges incident to nodes of $c$ that connect to the respective
side of the cut. Denote by $c_+$ the total ``outgoing'' capacity of cluster $c$
towards the root's side and by $c_-$ the ``incoming'' capacity.
\item Each cluster $c$ learns its ancestor clusters in the spanning tree of
$C$.
\item Observe that for a cluster $c$, an edge does contribute to $c_-$ if and
only if it connects to a node within its subtree $\cT_c$. From the previous
step, this information is known to one of the endpoints of the edge. We
communicate this and determine in each cluster $c$ the values $c_+$ and $c_-$ by
aggregation on its spanning tree.
\item Suppose $e\in \cT$ is the edge from cluster $c$ to its
parent. Using aggregation on the spanning tree of $C$, we compute
\begin{equation*}
|f_e| = \sum_{c'\in \cT_c} c'_+ - c'_-\,.
\end{equation*}
\end{compactenum}
\item Determine the index $i_0$ (as in Step 3 of the centralized routine). Given
that $\rload(e)$ for each $e\in \cT$ is locally known, this is
performed in $\tilde{O}(D)$ rounds using binary search in combination with
converge- and broadcasts on a BFS tree. We set $\cF:= \{e\in
\cT\,|\,\rload(e)>2^{i_0-1}\}$.
\item Define $P$ as the set of clusters incident to edges in $\cF$. A simple
broadcast on the cluster spanning trees makes membership known to all nodes of
each cluster $c\in P$.
\item Iteratively mark clusters $c\notin P$ with at most one unmarked
neighboring cluster, until this process stops. Add all unmarked clusters that
retain more than $2$ unmarked neighboring clusters to $P$.
\item For each path with endpoints in $P$ whose inner nodes are unmarked
clusters not in $P$, find the edge $e\in \cT \setminus \cF$ of minimal capacity
and add it to $\cF$. This disconnects any two clusters $c,c'\in P$, $c\neq c'$,
in $\cT\setminus \cF$.
\item Each component of $\cT\setminus \cF$ and the spanning trees of
clusters induce a spanning tree of the corresponding component of $G$. Each such
component is a new cluster. Make the identifier of the unique $c\in P$ of each
cluster known to its nodes and delete all edges between nodes in the cluster
that are not part of its spanning tree.
\end{compactenum}
If all trees have depth $\tilde{O}(\sqrt{n})$, all the above steps can be
completed in $\tilde{O}(\sqrt{n}+D)$ rounds. Clearly, the first $3$ stated
invariants are satisfied by the given construction. As mentioned earlier, it is
also straightforward to update the edge lengths, i.e., establish the fourth
invariant. Once the distribution on the current level of recursion is computed,
one can hence sample and then move on to the next level.

For the detailed description of the algorithm, the recursion, a formal statement
of \theoremref{theorem:sample}, and the respective proofs, we refer to
\sectionref{sec:jtrees}.

\input{clustergraph}

\input{sparsifier}

\input{spanning}

\input{jtrees}

\input{gradient}

\todo{Add conclusions/open problems/future work?}

\bibliographystyle{abbrv}
\bibliography{../flow}

\end{document}

%% file: flowmacs.tex
\usepackage[letterpaper,margin=1.0in]{geometry}
\usepackage{amsthm}

\newtheorem{theorem}{Theorem}[section]
\newtheorem{lemma}[theorem]{Lemma}
\newtheorem{corollary}[theorem]{Corollary}
\newtheorem{definition}{Definition}[section]

\newtheorem*{theoremRinner}{Theorem~\ref{\repthmref}}

\newenvironment{theoremR}[1]
  {\def\repthmref{#1}\theoremRinner (restated)}{\endtheoremRinner}
	\newenvironment{theoremRF}[1]
  {\def\repthmref{#1}\theoremRinner (restated and rephrased)}{\endtheoremRinner}

\usepackage{amsmath}
\usepackage{amssymb}
\usepackage{paralist}
\usepackage{wrapfig}
\usepackage{xspace}
\usepackage{graphicx}
\usepackage{epstopdf}
\usepackage{subfig}
\usepackage[framemethod=tikz]{mdframed}

\usepackage[colorlinks,linkcolor=blue,filecolor=blue,citecolor=blue,urlcolor=blue,pdfstartview=FitH,plainpages=false]{hyperref}
\usepackage[capitalize]{cleveref}
\usepackage{algorithm}
\usepackage[noend]{algpseudocode}

\usepackage{times}
\usepackage{enumerate}
\usepackage[disable,textsize=tiny]{todonotes}

\newcommand{\poly}{\operatorname{\text{{\rm poly}}}}
\newcommand{\polylog}{\operatorname{\text{{\rm polylog}}}}

\renewcommand{\stretch}{\ensuremath{\mathrm{stretch}}}

\newcommand{\namedref}[2]{\hyperref[#2]{#1~\ref*{#2}}}
\newcommand{\sectionref}[1]{\namedref{Section}{#1}}

\newcommand{\definitionref}[1]{\namedref{Definition}{#1}}
\newcommand{\theoremref}[1]{\namedref{Theorem}{#1}}
\newcommand{\figureref}[1]{\namedref{Figure}{#1}}
\newcommand{\lemmaref}[1]{\namedref{Lemma}{#1}}
\newcommand{\corollaryref}[1]{\namedref{Corollary}{#1}}

\newcommand{\equalityref}[1]{\hyperref[#1]{Equality~\eqref{#1}}}
\newcommand{\inequalityref}[1]{\hyperref[#1]{Inequality~\eqref{#1}}}
\newcommand{\algorithmref}[1]{\hyperref[#1]{Algorithm~\ref{#1}}}

\newcommand{\vf}{\boldsymbol{f}}
\newcommand{\vb}{\boldsymbol{b}}
\newcommand{\vy}{\boldsymbol{y}}
\newcommand{\cT}{\mathcal{T}}
\newcommand{\cF}{\mathcal{F}}
\newcommand{\cR}{\mathcal{R}}
\newcommand{\cC}{\mathcal{C}}
\newcommand{\cV}{\mathcal{V}}
\newcommand{\cG}{\mathcal{G}}
\newcommand{\cE}{\mathcal{E}}
\newcommand{\cL}{\mathcal{L}}
\newcommand{\cA}{\mathcal{A}}
\newcommand{\cP}{\mathcal{P}}
\newcommand{\cD}{\mathcal{D}}
\newcommand{\cH}{\mathcal{H}}

\newcommand{\cS}{\mathcal{S}}

\newcommand{\cJ}{\mathcal{J}}

\newcommand{\Congest}{\textsf{CONGEST}\xspace}
\newcommand{\set}[1]{\left\{#1\right\}}
\newcommand{\Cp}{\mathrm{cap}}
\newcommand{\Reals}{\mathbb{R}}
\newcommand{\Nats}{\mathbb{N}}

\newcommand{\eps}{\varepsilon}
\newcommand{\maxnorm}[1]{\left\|#1\right\|_\infty}

\newcommand{\Set}[1]{\left\{ #1 \right\}}
\newcommand{\E}{\mathbb{E}}
\newcommand{\TT}{\mathbb{T}}
\DeclareMathOperator{\Smax}{smax}

\DeclareMathOperator{\sgn}{sgn}

\DeclareMathOperator{\rload}{rload}

\makeatletter
\renewcommand{\paragraph}{%
  \@startsection{paragraph}{4}%
  {\z@}{1.2ex \@plus .5ex \@minus .1ex}{-.7em}%
  {\normalfont\normalsize\bfseries}%
}
\makeatother

\def\top{\intercal}

%% file: intro.tex
\section{Introduction}

Computing a maximum flow is a fundamental task in network
optimization. While the problem has a decades-old history rich with
developments and improvements in the sequential setting, little is
known in the distributed setting. In fact, prior to this work, the
best known distributed time complexity in the standard \Congest model
remained at the trivial bound of $O(m)$, which is the time needed to
collect the entire topology and solve the problem locally. For
undirected networks, this paper improves this unsatisfying state to
near-optimality:
\begin{theorem}\label{theorem:main_informal}
  On undirected weighted graphs, a $(1+\eps)$-approximation of a
  maximum $s$-$t$ flow can be computed in
  $(D+ \sqrt{n})\cdot n^{o(1)} \eps^{-3}$ rounds of the \Congest model
  with high probability.
\end{theorem}
This round complexity almost matches the $\tilde{\Omega}(D+ \sqrt{n})$
lower bound of Das Sarma et al.~\cite{DasSarma-11}, which holds for
any non-trivial approximation.  Before we proceed, let us formalize
the model and the problem.

\todo{I guess we need a bit more exposition in the introduction of the journal
version.}

\subsection{Model and Problem}
\paragraph{Model.}
We use the standard \Congest model of synchronous
computation~\cite{Peleg:book}.  We are given a simple, connected,
weighted graph $G=(V,E,\Cp)$, where $\Cp:E\to \Nats$,
$\Cp(e)\in \poly n$, are the edge capacities.\footnote{As merely an
  approximate flow is required, we can reduce the general case to this
  setting in $\tilde{O}((\sqrt{n}+D)\log C)$ rounds, where $C$ is an
  upper bound on the ratio between the largest and smallest capacity.}
\todo{Journal version: Add a lemma and/or a short appendix section
  explaining this footnote somewhere.}  By $D$, we denote the (hop)
diameter of $G$. Each of the $n:=|V|$ nodes hosts a processor with a
unique identifier of $O(\log n)$ bits, and over each of the $m:=|E|$
edges $O(\log n)$ bits can be sent in each synchronous round of
communication; we assume that nodes have access to infinite strings of
independent unbiased random bits. We say that an event occurs
\emph{with high probability} (w.h.p.), if it happens with probability
$1-n^{-c}$ for any desired constant $c>0$ specified
upfront.\footnote{Taking the union bound over polynomially many events
  does not affect this property. We will use this fact frequently and
  implicitly throughout the paper.} Initially, each node only knows
its identifier, its incident edges, and their capacities.

\paragraph{Problem.}
We fix an arbitrary orientation of the edges. In the following, we write
$(u,v)\in E$ if $\{u,v\}\in E$ is directed from $u$ to $v$. An instance of the
(single-commodity) \emph{max flow} problem is given by, in addition to
specifying $G$, designating a source $s\in V$ and a sink $t\in V$. A (feasible)
\emph{flow} is a vector $\vf \in \Reals^E$ satisfying:
\begin{compactenum}
\item capacity constraints (edges): $\forall e\in E:~|f_e|\leq 
\Cp(e)$\,;
\item preservation constraints (nodes): $\forall u\in 
V\setminus\{s,t\}:~\sum_{(u,v)\in E}f_e - \sum_{(v,u)\in
E}f_e = 0$\,; and
\item $\sum_{(s,u)\in E}f_e - \sum_{(u,s)\in E}f_e = - \sum_{(t,u)\in
E}f_e + \sum_{(u,t)\in E}f_e = F\in \Reals$\,.
\end{compactenum}
Here, $F$ is the \emph{value} of $\vf$. A \emph{max flow} is a flow of
maximum value. For $\eps>0$, a \emph{$(1+\eps)$-approximate max flow}
is a flow whose value is by at most a factor $1+\eps$ smaller than
that of a max flow. In this work, we focus on solving the problem of
finding a $(1+\eps)$-approximate max flow in the above model, where it
suffices that each node $u$ learns $f_e$ for its incident edges
$\{u,v\}\in E$.

\subsection{Related Work}\label{sec:related}

\todo{Need to extend the related work for journal version. (E.g., a bit on
centralized algorithms leading to Sherman, cite KLOS and Peng's improvement to
$\tilde{O}(m)$.)}

Network flow, being one of the canonical and most useful optimization
problems, has been the target of innumerable research efforts since
the 1930s~\cite{schrijver} (see, e.g., the classic book
\cite{AhujaMO-93} and the recent survey~\cite{GT14}).  For the
general, directed case, the fastest known sequential algorithm is by
Goldberg and Rao and it solves the max flow problem in time
$\tilde{O}(m\cdot \min\set{m^{1/2},n^{1/2}})$. Particularly relevant
from the point of view of the present paper are recent efforts to
obtain fast algorithms to compute (approximate) max flow solutions in
the undirected case. Using the graph sparsification technique of
Bencz\'ur and Karger \cite{BK15}, any graph can be partitioned into
$k=\tilde{O}(m\eps^2/n)$ sparse graphs with $\tilde{O}(n/\eps^2)$
edges such that the max flow problem can be approximately solved by
combining max flow solutions for each of these sparse graphs. Using
the algorithm of Goldberg and Rao, this results in an algorithm with
running time $\tilde{O}(mn^{1/2})$. In \cite{CKMST11}, Christiano et
al.\ improved this running time to $\tilde{O}(mn^{1/3})$ by applying
the almost linear-time Laplacian solver of Spielman and Teng
\cite{ST06} to iteratively minimize a softmax approximation of the
edge congestions. Kelner et al.~\cite{KLOS14} and
Sherman~\cite{Sherman13} independently published two algorithms which
allow to compute a $(1+\eps)$-approximation to an undirected max flow
problem in time almost linear in $m$. 
Finally, Peng~\cite{Peng14} proved the first running time in $O(m \polylog(n))$.

However, to the dismay of many, and despite the fact that the word 
``network'' even appears in the problem's name, only little progress 
was made over the years from the standpoint of distributed algorithms.
For example, Goldberg and Tarjan's push-relabel algorithm, which is 
very local and simple to implement 
in the \Congest model, requires $\Omega(n^2)$ rounds to converge, where 
$n$ is 
the number of nodes. This is very 
disappointing, because in the \Congest model, any problem whose input 
and output can be encoded with $O(\log n)$ bits per edge, can be 
trivially solved 
in $O(m)$ rounds, where $m$ is the number of edges, by collecting all 
input at a single node, 
solving it there, and distributing the results back.

Early attempts focused, as customary in those days, on reducing the
number of messages in asynchronous executions. 
For example, Segall~\cite{Segall82} 
gives an $O(nm^2)$-messages, 
$O(n^2m)$-time algorithm for exact max flow, and Gafni and Marberg~\cite{MG87} give an an algorithm whose message and time complexities
are $O(n^2m^{1/2})$.
Awerbuch has attacked the problem repeatedly with the following 
results. In an early work~\cite{A85:SynchronizerUse} he adapts Dinic's
centralized algorithm using a synchronizer, giving rise to an algorithm whose
time and message complexities are $O(n^3)$.
With Leighton, in \cite{AL-94} they give an algorithm for 
solving multicommodity flow  approximately in $O(\ell m\log m)$ rounds, 
where $\ell<n$ is the length of the longest flow path. Later he 
considers the model where each flow path (variable) has an ``agent'' 
which can find the congestion of all links on its path in 
\emph{constant 
time}. In this model, he shows with Khandekar \cite{AK-09}  how 
to approximate any positive LP (max flow with given routes 
included) to within 
$(1-\epsilon)$ in time polynomial in $\log(mn A_{\max}/\epsilon)$
(here $n$ is the number of variables, which is at least the
number of paths considered). The 
same model is used with Khandekar and Rao in \cite{AKR12}, where 
they show how to approximate multicommodity flow to within 
$(1-\epsilon)$ in 
$O(\ell\log n)$ rounds. Using a straightforward implementation of this 
algorithm in the \Congest model results in an $\tilde O(n^2)$-time 
algorithm.

Thus, up to the current paper, there was no distributed 
implementation 
of a max-flow algorithm which always requires a sub-quadratic number of 
rounds. Even an $O(n)$-time algorithm would have been considered a significant 
improvement, even for the $0/1$ capacity case.

\subsection{Organization of this Article}

Our result builds heavily on a few major breakthroughs in the
understanding of max flow in the centralized setting, most notably the
almost linear-time approximation algorithm for the undirected max flow
problem by Sherman \cite{Sherman13}, as well as a few other
contributions.
We first give an overview of the key concepts in Sections
\ref{sec:overview}--\ref{sec:approximator}. We carefully revisit
Sherman's approach~\cite{Sherman13} and the main building blocks he
relies on in \sectionref{sec:overview}. This sets the stage for
shedding light on the challenges that must be overcome for its
distributed implementation and presentation of our results
in \sectionref{sec:contribution}.  There, we also provide a top-level
view of the components of the algorithm, alongside pointers to the
detailed proofs in Sections \ref{sec:cluster}--\ref{sec:gradient}
showing that we can implement each of them by efficient distributed
algorithms. In \sectionref{sec:approximator}, we outline the
distributed construction of an \emph{$n^{o(1)}$-congestion
  approximator}, which is our key technical contribution; the role of
a congestion approximator is to \emph{estimate} the congestion induced
by optimally routing an arbitrary demand vector very quickly, which
lies at the heart of the algorithm. All the details of our distributed
algorithm and all the proofs appear in Sections
\ref{sec:cluster}--\ref{sec:gradient}.


%% file: clustergraph.tex
\section{Cluster Graphs}\label{sec:cluster}

On several levels, our distributed congestion approximator
construction is done in a hierarchical way. As a consequence many of
the distributed computations used by our algorithm have to be run on a
graph induced by clusters of the network graph. In order to be able to
deal with such cluster graphs in a systematic way, we formally define
cluster graphs and we describe how to simulate distributed
computations on a cluster graph by running a distributed algorithm on
the underlying network graph.

\begin{definition}[Distributed Cluster Graph]\label{def:clustergraph}
  Given a $n$-network graph $G=(V,E)$, a distributed $N$-node cluster
  graph $\cG=(\cV,\cE,\cL,\mathfrak{T},\psi)$ of size $n$ is
  defined by a set of $N$ clusters $\cV=\set{S_1,\dots,S_N}$
  partitioning the vertex set $V$, a set (or multi-set) of edges
  $\cE\subseteq \binom{\cV}{2}$, a set of cluster leaders $\cL$, a set
  of cluster trees $\mathfrak{T}$, as well as a function $\psi$ that
  maps the edges $\cE$ of the cluster graph to edges in $E$. Formally,
  the tuple $(\cV,\cE,\cL,\mathfrak{T},\psi)$ has to satisfy the
  following conditions.
  \begin{enumerate}[(I)]
  \item The clusters $\cV = (S_1, \dots, S_N)$ form a partition of the
    set of vertices $V$ of the network graph, i.e., $\forall i\in [N]:
    S_i\subseteq V$, $\forall 1\leq i<j\leq N: S_i\cap S_j=\emptyset$,
    and $\bigcup_{i=1}^N S_i = V$.
  \item For each cluster $S_i$, $|S_i\cap \cL|=1$. Hence, each cluster
    has exactly one cluster leader $\ell_i\in \cL\cap S_i$. The ID of
    the node $\ell_i$ also serves as the ID of the cluster $S_i$ and
    for the purpose of distributed computations, we assume that all
    nodes $v\in S_i$ know the cluster ID and the size $n_i:=|S_i|$ of
    their cluster $S_i$.
  \item Each cluster tree $T_i=(S_i,E_i)$ is a rooted spanning tree of
    the subgraph $G[S_i]$ of $G$ induced by $S_i$. The root of $T_i$
    is the cluster leader $\ell_i\in S_i\cap \cL$. We assume that each
    node of $u\in S_i\setminus \set{\ell_i}$ knows its parent node
    $v\in S_i$ in the tree $T_i$.
  \item The function $psi:\cE\to E$ maps each edge $\set{S_i,S_j}\in
    \cE$ to an (actual) edge $\set{v_i,v_j}\in E$ connecting the
    clusters $S_i$ and $S_j$, i.e., it holds that $v_i\in S_i$ and
    $v_j\in S_j$. The two nodes $v_i$ and $v_j$ know that the edge
    $\set{v_i,v_j}$ is used to connect clusters $S_i$ and $S_j$. If
    the cluster graph is weighted, the two nodes $v_i$ and $v_j$ also
    know the weight of the edge $\set{S_i,S_j}$.
  \end{enumerate}
\end{definition}

\noindent Note that (III) in particular implies that the subgraph of
$G$ induced by each cluster $S_i$ is connected. When dealing with a
concrete distributed cluster graph $\cG$, we use $\cG_\cV$, $\cG_\cE$,
$\cG_\cL$, $\cG_\mathfrak{T}$, and $\cG_psi$ to denote the corresponding sets
of clusters, edges, etc. Further, when only arguing about the cluster
graph and not its mapping to $G$, we only use the pair $(\cV,\cE)$ to
refer to it. In the following, we say that a cluster $S\in \cV$
\emph{knows} something if all nodes $v\in S$ know it. That is, e.g.,
the last part of condition (II) says that every cluster knows its ID
and its size.

We next define a weak version of the (synchronous) \Congest model and
we show that algorithms in this model can be efficiently simulated in
distributed cluster graphs.

\begin{definition}[$B$-Bounded Space \Congest Model]\label{def:boundedspacecongest}
  Let $B=\Omega(\log n)$ be a  given parameter. The \emph{$B$-Bounded Space
    \Congest 
  model} is a computation model which restricts the $\Congest(B)$
model by requiring, for
    any $d\ge0$, that each step of a node $v$ in  the $B$-Bounded Space
can be emulated in $O(d)$ rounds by any tree $T(v)$ of depth $d$, where
the edges incident on $v$ are incident on nodes of $T(v)$ in the emulation.
\end{definition}

The definition of the $B$-Bounded Space model is  directed
toward emulation. The definition immediately implies that if each node
is emulated by a tree, then emulating a global step in time
proportional to the maximal tree depth (which could be $\Omega(n)$) is
trivial. However, the 
following lemma shows that this can actually be done in time
$O(D + \sqrt{n})$. 
\begin{lemma}\label{lemma:clustergraphsimulation}
  Given an underlying $n$-node graph $G=(V,E)$ and a cluster
  graph $\cG=(\cV,\cE,\cL,\mathfrak{T},psi)$, a $t$-round distributed
  algorithm $\cA$ in $\cG$ in the $B$-bounded space \Congest model can
  be simulated in the (ordinary) \Congest model in $G$ with messages
  of size at most $B$ in $O\big((D + \sqrt{n})\cdot t\big)$ rounds,
  where $D$ is the diameter of $G$.
\end{lemma}
\begin{proof}[Proof Sketch]
  We assume that we are given a global BFS tree of $G$. If such a BFS
  tree is not available, it can be computed in $O(D)$ rounds in the
  $\Congest$ model.  We simulate the algorithm $\cA$ in a
  round-by-round manner. Consider the end of the simulation of round
  $r-1$ and assume that in each cluster $S_i\in \cV$, the leader node
  $\ell_i$ knows the message $M_i$ to be sent in round $r$. (For
  $r=1$, we assume that this is true at the beginning of the
  simulation.)

  To start the simulation of round $r$, we first make sure that for
  every $S_i\in \cV$, every node $v\in S_i$ knows the message $M_i$ to
  be sent to the neighbors in round $r$. In clusters $S_i$ of size at
  most $\sqrt{n}$, this can be done in at most $\sqrt{n}$ rounds by
  broadcasting $M_i$ on the spanning tree $T_i$ of $G[S_i]$. For
  larger clusters, we use the global BFS tree to disseminate the
  information. We first send all the messages $M_i$ of clusters $S_i$
  of size larger than $\sqrt{n}$ to the root of the global BFS
  tree. Because the BFS tree has radius at most $D$ and because there
  are at most $\sqrt{n}$ clusters of size larger than $\sqrt{n}$, this
  can be done in $D+\sqrt{n}$ rounds (using pipelining). Now, in
  another $D+\sqrt{n}$ rounds, all these messages can be broadcast to
  all nodes of $G$ (and thus also to the nodes of the clusters that
  need to know them).

  Now, for every two clusters $S_i$ and $S_j$ such that
  $\set{S_i,S_j}\in \cE$, let $\set{u_i,u_j}=psi(\set{S_i,S_j})$ be
  the physical edge connecting $S_i$ and $S_j$. The node $u_i\in S_i$
  sends the message $M_i$ to $u_j$ and the node $u_j\in S_j$ sends
  $M_j$ to $u_i$. This step can be done in a single round. Now, in
  each cluster $S_i$, each incoming message of round $r$ is known by
  one node in $S_i$ and we need to aggregate these messages in order
  to compute the outgoing message of each cluster. In clusters of size
  at most $\sqrt{n}$, this can again be done locally inside the
  cluster (by \definitionref{def:boundedspacecongest}). Also, for the at
  most $\sqrt{n}$ clusters of size larger than $\sqrt{n}$, we again
  use the global BFS tree. Since in a tree of depth $D$, $k$
  independent convergecasts of broadcasts can be done in time $D+k$,
  the messages of the large clusters can be computed in disseminated
  to the cluster leaders in time $O(D+\sqrt{n})$.
\end{proof}


%% file: sparsifier.tex
\section{Distributed Construction of Cut 
Sparsifiers}\label{sec:sparsifiers}

\begin{figure}[tb]
    \centering
\begin{mdframed}
\small
\noindent
\begin{compactenum}
\item $R_0:=\Set{\Set{v}\mid v\in V}$.
\item For $i:=1$ to $\log N$ do:
\begin{compactenum}
\item \label{bs-rand} 
\emph{Mark} each cluster of $R_{i-1}$ independently
with probability $\frac12$; let $R_{i}:=\Set{S\text{ is
    marked}\mid S\in R_{i-1}}$.
\item \label{bs-first} If $v\in S$ for some $S\in R_{i-1}\setminus R_{i}$:
\begin{compactenum}
\item 
Define $Q_v$ to be the set of edges that consists of the lightest edge
from $v$ to each of the clusters in $R_i$ $v$ is adjacent to.
\item If $v$ has no neighbor in a cluster in $R_{i}$, 
then $v$ adds to the
spanner all edges in $Q_v$.
\item Otherwise, let $u$ be the closest neighbor of $v$ in a
marked cluster. Then 
\begin{compactitem}
\item $v$ {joins} the cluster of $u$ (i.e., if $u$ is in cluster $S'\in
  R_i$,
then $S':=S'\cup \{v\}$).
\item $v$ adds to the spanner the edge
$\{v,u\}$, and also all edges $\{v,w\}\in Q_v$ with $W(v,w)<W(v,u)$
(breaking ties by ID).
\end{compactitem}
\end{compactenum}
\end{compactenum}
\item \label{bs-final}
Each node $v$ adds, for each cluster $S\in R_{\log N}$ it is adjacent to, the lightest
edge connecting it to $S$.
\end{compactenum}
\end{mdframed}
    \caption{\small The BS algorithm for $O(\log N)$ spanner 
    construction
      given an $N$-node weighted graph $G=(V,E,W)$. The output is a
      subset of $E$.}
    \label{fig-bs}
  \end{figure}

\begin{lemma}\label{lemma:sparsification}
  In a weighted $N$-node distributed cluster graph of size $n$, for
  any $\eps>0$, it is possible to compute a spectral
  $(1+\eps)$-sparsifier with $O(N\cdot (\eps^{-1}\cdot\log N)^{O(1)})$
  edges (w.h.p.) in the \Congest model in time
  $O\big((D+\sqrt{n})\cdot (\eps^{-1}\cdot\log N)^{O(1)}\big)$. When
  the algorithm terminates, each of the edges of the sparsifier is
  directed such that the out-degree of each cluster is upper bounded
  by $O( (\eps^{-1}\cdot\log N)^{O(1)})$ and such that each cluster
  knows all its outgoing edges.
\end{lemma}

\begin{proof}
We prove the lemma using the algorithm \textsc{ParallelSparsify} of
Koutis~\cite{Koutis14}, and then orient the edges. Koutis's algorithm relies on
the $O(\log n)$-stretch spanner construction algorithm of Baswana and
Sen~\cite{BaswanaSen}, which we henceforth refer to as BS. See
\figureref{fig-bs} for a description of the BS algorithm.

We start by showing how to emulate a step of node $v$ in BS by a
depth-$d$ tree $T_v$ in 
$O(d+\log N)$ 
time (w.h.p.). We may 
assume w.l.o.g.\ the existence of a root in each tree (because we can
select one in $O(d)$ time); Step \ref{bs-rand} is carried out by the
root and the result is broadcast over the tree. For Step \ref{bs-first},
we note that w.h.p., $|Q_v|=O(\log N)$ and hence making it known to
all node $T_v$ takes $O(d+\log N)$ time using standard
convergecast-broadcast. Step \ref{bs-final} is straightforward given
that each node knows its cluster.

Next, given a tree $T(v)$ for each node, we assume
 that the depth of all trees is at least $c\log N$ for
some appropriate constant  $c>0$. If this assumption does not hold  we
extend $T(v)$ with 
a dummy path: clearly $T(v)$ can emulate the extended tree without any
slowdown. However this extension may increase the number of nodes by
an $O(\log N)$ factor. 
Now, under this assumption and the emulation above,  we may apply
\lemmaref{lemma:clustergraphsimulation} to 
conclude that BS can be executed in time $O((D + \sqrt{N\log N})\log N)$.

Going back to the algorithm of Koutis \cite{Koutis14}, we note that it
consists of  $({\log n/\varepsilon})^{O(1)}$ invocations of BS,
and some  independent random selection and 
reweighting of edges. The former is discussed above, and the latter is
trivial to emulate locally.

Finally, for edge orientation, we give a little algorithm that, given
an  $N$-node, $D$-diameter graph with 
average degree $d_{\text{av}}$, orients all edges such that the
out degree of all nodes is $O(d_{\text{av}})$. The algorithm runs
in $O(D+\log n)$ steps in the 
space-bounded \Congest model. The algorithm is as follows. First
compute the average degree in $O(D)$ time, and then repeat the
following procedure $\log n$ times at each node $v$:
\begin{compactitem}
\item If the number of unoriented edges incident on $v$ is less than
  $2d_{\text{av}}$, then $v$ orients all unoriented edges outward, informs its
  neighbors, and halts.
\end{compactitem}
The correctness of the procedure follows from the fact that throughout
the execution, at most
half the non-halted nodes have degree larger than $2d_{\text{av}}$. The lemma
now follows from the fact that the graph generated by Koutis' algorithm has
average degree $(\frac{\log N}{\varepsilon})^{O(1)}$.
\end{proof}

%% file: spanning.tex
\section{Distributed Construction of Low Average-Stretch Spanning 
Trees}\label{sec:spanning}

\begin{theoremRF}{theorem:spanning}
Suppose $H$ is a multigraph obtained from $G$ by assigning arbitrary
edge lengths in $2^{n^{o(1)}}$ to the edges of $G$ (known to incident nodes) and
performing an arbitrary sequence of contractions. Then we can compute a rooted
spanning tree of $H$ of expected stretch $2^{O(\sqrt{\log n \log \log n})}$
within $(\sqrt{n}+D)n^{o(1)}$ rounds, where the edges of the tree in $H$ and
their orientation is locally known to the endpoints of the corresponding edges
in $G$.
\end{theoremRF}

\begin{figure}[tb]
\centering
\begin{compactenum}
\item $G^1=(V^1,E^1):=G$; $\cC:=\emptyset$.
\item For $t=1$ to $2\log N$ do:
  \begin{compactenum}
  \item Let $S^t$ random subset of $V^t$ of $12 \frac{2^{t/2}}{n}|V^t|$
    nodes; if $V^t$ is smaller than $\frac{n}{12\cdot2^{t/2}}$,
    $S^t:=V^t$.
  \item $\cC:=\cC\cup\Set{\Set{s}\mid s\in S^t}$.
  \item Each $s\in S^t$ draws a random delay $\delta_s^t$
    uniformly from $[0,\lfloor\rho/(2\log N)\rfloor]$.
  \item Each $s\in S^t$ waits $\delta_s^t$ rounds and then
    initiates a BFS for $\rho(1-\frac{t-1}{2\log N}) -\delta_s^t$ rounds in
    $G^t$.
  \item A node covered by a BFS is added to cluster $C_s$, where $s$
    is the source of the first BFS to visit it, breaking ties by ID.
  \item $V^{t+1}:=V\setminus\Set{v\mid v\in C\text{ for some
    }C\in\cC}$;
    $G^{t+1}:=G^t[V^{t+1}]$.
  \end{compactenum}
\end{compactenum}
\caption{Algorithm SplitGraph. The input is an unweighted
graph $G=(V,E)$ and a target radius $\rho$.}
\label{fig-split}
  \end{figure}

\begin{proof}
  We follow \cite{blelloch14}: the high-level algorithm is by Alon et
  el.~\cite{alon1995graph}, which uses Algorithm Partition (of
  \cite{blelloch14}) for unweighted graphs. We describe the algorithm
  bottom-up. The main 
  component in Algorithm Partition is Algorithm SplitGraph, reproduced
  in \figureref{fig-split}. 
  The basic action of Algorithm SplitGraph is growing BFS trees, an 
  action in
  which emulating a single node by a tree is trivial. In SplitGraph we
  may have contending BFS growths, but note that if two
  or more BFS traversals collide, only the winning ID needs to
  proceed, and hence there are no collisions because no edge needs to
  carry more than a single BFS traversal in each direction.  Regarding
  the tree construction, we first note that the BFS growth
  naturally creates a spanning tree for each cluster. Moreover, we can
  make all nodes know the complete path to the root of their
  respective cluster
  in additional $O(\rho)$ steps, by letting each node send its
  $i^{th}$ ancestor to all its children in round $i$.  The running time of Algorithm
  SplitGraph is clearly $O(\rho\log N)$ in the bounded-space \Congest
  model, and therefore, using 
  \lemmaref{lemma:clustergraphsimulation}, we conclude that we can run
  SplitGraph in time $O(\rho\log N(D+\sqrt{N}))$ in the \Congest model.

  Algorithm SplitGraph is called by Algorithm Partition, whose input
  is an unweighted graph with an arbitrary partition of the edges into
  $K$ classes. Algorithm Partition 
  applies Algorithm 
  SplitGraph disregarding classes, and then checks whether there
  exists a class where too many edges were split in different
  clusters. If there is such an 
  over-split class, the algorithm is restarted.  We can implement each
  checking and restart in the \Congest model in $O(D+K)$ time using a 
  global BFS
  tree. Since the number
  of restarts is bounded by $O(\log N)$ w.h.p.~\cite{blelloch14}, and 
  in our implementation we shall have $K=O(\sqrt{N})$, the
  overall time for running Algorithm Partition is $O(\rho\log^2
  N(D+\sqrt{N}))$ in the \Congest model. 

The outermost algorithm is the one by Alon et
al.~\cite{alon1995graph}, whose input is a \emph{weighted} graph. The
algorithm first partitions the edges into $O(\sqrt{\log N})$ classes by
weight, where class $E_i$ contains all edges whose weight is in
$[z^{i-1},z^i)$ for a certain value 
$z=\tilde\Theta(2^{\sqrt{6\log N\cdot\log\log N}})$. Then 
the algorithm proceeds in iterations until the graph is a single node,
where iteration $j$ is as follows. 
\begin{compactenum}
\item Call Algorithm Partition with edges $E_1,\ldots,E_j$ and target radius
$\rho=z/4$. Obtain clusters $\Set{C_i}$.
\item Output  a BFS tree
for each cluster $C_i$.
\item Contract each
resulting cluster $C_i$ to a single node.  Remove all self loops, but
leave parallel edges in place. The resulting multigraph, augmented with
edge class 
$E_{j+1}$, is the input to iteration $j+1$.
\end{compactenum}
For the distributed implementation, note that edge contraction is
trivial given that the endpoints know the 
identity of the cluster they belong to, and that edge classification
is purely local given $z$ (which can be communicated to all in $O(D)$
time units). It can be shown~\cite{blelloch14} that w.h.p., the number of
iterations is $O({\log\Delta/\sqrt{\log N \log\log N}})$, and hence the running time
of the algorithm is 
$O(\rho\log\Delta\log^{O(1)}N(D+\sqrt N))=
\log\Delta\cdot2^{O(\sqrt{\log N\log\log N})}$ because
$\rho=\tilde\Theta(2^{\sqrt{6\log N\cdot\log\log N}})$. 

The claimed stretch follows from \cite{blelloch14}.
\end{proof}

%% file: jtrees.tex
\section{Distributed Construction of \texorpdfstring{$j$}{j}-Trees}
\label{sec:jtrees}

From the distributed implementation point of view, the core technical
challenge is to efficiently compute a congestion approximator in a
distributed way. As already pointed out, the congestion approximator
is constructed based on applying the $j$-tree construction of Madry
\cite{Madry10} recursively. In the following, we review Madry's
construction and we show how to implement an adapted version of it in
a distributed network. The main objective of the construction is to
approximate the flow structure of a given graph by a distribution of
graphs from a simpler class of graphs (i.e., $j$-trees). Formally,
the similarity of the flow structure of two graphs is captured by the following
definition from \cite{Madry10}. 

\begin{definition}[Graph Embeddability]\cite{Madry10}
  We are given $\beta\geq 1$ and two (multi)-graphs $G=(V,E,\Cp)$ and
  $G'=(V,E',\Cp')$ on the same set of nodes and with edge capacities
  $\Cp(e)$ and $\Cp'(e')$ for edges $e\in E$ and $e'\in E'$.  We say that
  graph $G$ is $\beta$-embeddable into $G'$ if there exists a
  \emph{multicommodity flow} $\vf'=(\vf'_e)_{e\in E}$ such that for
  every edge $e\in E$ of $G$ connecting nodes $u$ and $v$, $\vf'_e$ is a
  flow on $(V,E',\beta\, \Cp')$ that routes $\Cp(e)$ units of flow between $u$
  and $v$, and for every edge $e'\in E'$ of $G'$, it holds that
  $|\vf'(e')|:=\sum_{e\in E}|(\vf'_e)(e')|\leq\beta\, \Cp'(e')$.
\end{definition}

\noindent Intuitively, a graph $G$ is $\beta$-embeddable into a graph
$G'$, if for every (multicommodity) flow problem, there is a solution
in $G'$ such that the maximum relative congestion of all edges is by at most a
factor $\beta$ larger than for the optimal solution in $G$. As a generalization
of the cut-based graph decompositions of R\"acke \cite{raecke08}, Madry defines
the notion of an $(\alpha,\mathbb{G})$-decomposition.

\begin{definition}[$(\alpha,\mathbb{G})$-Decomposition~\cite{Madry10}]
  Given a (multi-)graph $G=(V,E,\Cp)$ and a family $\mathbb{G}$ of graphs
  on the nodes $V$, an $(\alpha,\mathbb{G})$ of $G$ is a set of
  pairs $\set{(\lambda_i,G_i)}_{i\in I}$ satisfying that:
  \begin{compactitem}
  \item $\forall i\in I:~\lambda_i>0$;
  \item $\sum_{i\in I} \lambda_i = 1$;
  \item $\forall i\in I:~G_i=(V,E_{i},\Cp_{i})$ is a graph in $\mathbb{G}$;
  \item $\forall i\in I:~G$ is $1$-embeddable into $G$; and
  \item the graph defined by the convex combination\footnote{The sum of two
  weighted graphs $G_1=(V,E_1,\Cp_1)$ and $G_2=(V,E_2,\Cp_2)$ is defined as
$G_1+G_2=(V,E_1\cup E_2, \Cp_{12})$, where for each edge $e\in E_1\cup
E_2$, $\Cp_{12}(e)$ is defined as $\Cp_{12}(e)=\Cp_1(e)+\Cp_2(e)$ if $e\in
E_1\cap E_2$ and $\Cp_{12}(e)=\Cp_i(e)$ if $e\in E_i\setminus E_{3-i}$ for
$i\in \set{1,2}$.} $\sum_{i\in I}
  \lambda_i\cdot G_i$ is $\alpha$-embeddable into $G$.
  \end{compactitem}
  In words, $\set{(\lambda_i,G_i)}_{i\in I}$ is a distribution on $I$ graphs
  from $\mathbb{G}$, each of which can be $1$-embedded into $G$, such that the
  distribution $\alpha$-embeds into $G$.
\end{definition}

\noindent Observe that such a decomposition can form the basis for a good
congestion approximator: $1$-embeddability of each $G_i$ into $G$
guarantees that congestion is never overestimated, and the embeddability of the
convex combination ensures when sampling from the distribution, the expected
factor by which we underestimate congestion on a cut is at most $\alpha$. Our
goals are now to choose $\mathbb{G}$ and the distribution such that
\begin{compactitem}
\item $\alpha$ is small,
\item we can construct the distribution efficiently, and
\item we can evaluate the induced congestion when routing demand optimally on
a graph from the distribution efficiently.
\end{compactitem}

\paragraph{The Plan.}
Let $G=(V,E,\Cp)$ be a weighted (multi-)graph, $0\leq j\leq |V|$ be an
integer and let $\mathbb{J}$ be the family of $j$-trees over the node
set $V$. In \cite{Madry10}, it is shown that based on a protocol for
computing spanning trees with average stretch $\alpha$, there exists
an $(\alpha,\mathbb{J})$-decomposition of $G$. This is shown in
several steps. It is first shown that a sparse
$(\alpha,\mathbb{H})$-decomposition exists for a graph family
$\mathbb{H}$ which contains graphs that are closer to the original
graph $G$ and it is then shown that every graph $\mathcal{H}\in
\mathbb{H}$ can be $O(1)$-embedded into a $j$-tree and vice versa.

As described, we have to apply the $j$-tree construction recursively
to the core graph. Each node in the core graph is represented by a set
of nodes (a cluster) in the network graph. On the network graph, the
core graph therefore corresponds to a graph between clusters of
nodes. We therefore have to be able to apply the $j$-tree construction
on a cluster graph. As we will see, we can construct $j$-trees such
that whenever two nodes $u$ and $v$ of the core are connected by a
(virtual) edge, there also is a physical edge between the two trees
(i.e., clusters of nodes) corresponding to $u$ and $v$. Throughout our
algorithm, we can therefore work with a cluster multigraph such that
a) the induced graph of each cluster is connected and b) for every
edge between two clusters $c$ and $c'$, there are nodes $u\in
c$ and $v\in c'$ such that $u$ and $v$ are connected by an
edge in the underlying network graph. For doing distributed
computations, we assume that each cluster has a leader and that every
node knows the ID of the leader and also its parent in a rooted
spanning tree which is rooted at the leader. In \sectionref{sec:cluster}, we
give a precise definition of a \emph{distributed cluster graph} and we show that
several basic algorithms that we use as building blocks can be run efficiently
in distributed cluster graphs.

In the following, we go through Madry's $j$-tree construction
step-by-step and describe how to adapt it so that we can implement it
efficiently on a distributed cluster graph (i.e., in the \Congest model in
the underlying network graph).

\subsection{Low-Stretch Spanning Trees}
\label{subsec:LSST}

In the following, we consider the computation of the
$(\alpha,\mathbb{J})$-decomposition of some core graph. Assume that
the core graph is given as a distributed cluster graph
$\cG=(\cV,\cE,\Cp)$, where for each edge $e\in \cE$, $\Cp(e)$ is the
capacity of $e$. As the time complexity of some of the steps for
computing an $(\alpha,\mathbb{J})$-decomposition of $\cG$ depend on
the number of edges of $\cG$, as a first step, we sparsify $\cG$. In
\lemmaref{lemma:sparsification}, it
is shown that in $O\big((D+\sqrt{n})\cdot \polylog n\big)$ rounds, it
is possible to compute an $(1+1/\polylog n)$-spectral sparsifier of
$\cG$ with at most $O(|\cV|\cdot \polylog n)$ edges. Further, for each
edge $\set{c,c'}\in \cE$ of the sparsifier one of the nodes of the
edge manages the edge. As in general $\cG$ is a cluster graph,
$c$ and $c'$ are clusters of physical nodes and an edge connecting
clusters $c$ and $c'$ is represented by a physical edge $\set{u,v}$
for two (network) nodes $u\in c$ and $v\in c'$. We will maintain that every
pair of nodes $u\in c$ and $v\in c'$ needs to represent at most one
edge between $c$ and $c'$ in $\cG$. The two nodes $u$ and $v$ know
about the edge between $c$ and $c'$ and its capacity.

In the following, we assume that $\cG$ is the graph after
sparsification. If the number of nodes $|\cV|$ of $\cG$ is less than
$n^{1/2+o(1)}$, using a global BFS tree of the network graph, the
whole structure of $\cG$ can be collected in $O(D+|\cV|\polylog n)=O(D
+ n^{1/2+o(1)})$ rounds. In that case, we can therefore perform
all remaining operations locally at the nodes. Consequently, we
will henceforth assume that $|\cV|\geq n^{1/2+o(1)}$.

During the construction of the $(\alpha,\mathbb{J})$-decomposition of
$\cG$, each edge $e\in \cE$ is assigned a length $\ell(e)$. At the
beginning $\ell(e)$ is proportional to $1/\Cp(e)$ and before adding
each $j$-tree, $\ell(e)$ is adapted for each edge. As the first step
of constructing each $j$-tree in the decomposition, Madry computes a
spanning tree $\cT$ of $\cG$ for which it holds that
\begin{equation}\label{eq:weightedavgstretch}
\sum_{e=\set{c,c'}\in \cE} \!\!\!\!d_{\cT}(c,c')\cdot \Cp(e) 
\ =\ 
\!\!\!\!\sum_{e=\set{c,c'}\in \cE} \!\!\!\!\stretch_{\cT}(e)\cdot
\ell(e)\cdot \Cp(e) 
\ \leq\ 
\delta\alpha \cdot \!\!\!\!\!\!\!\sum_{e=\set{c,c'}\in\cE} \!\!\!\!\ell(e)\cdot
\Cp(e)
\end{equation}
for a sufficiently small positive constant $\delta$. In the above
expression, $d_{\cT}(u,v)$ denotes the sum of edge lengths on the path
between $u$ and $v$ on $\cT$. Hence, $\cT$ is a spanning tree with
a bounded weighted average stretch. Such a spanning tree can be
computed by computing an (unweighted) low average stretch spanning tree
for a multigraph $\tilde{\cG}$ which is obtained from $\cG$ by
(logically) replacing some of the edges of $\cG$ with multiple copies of the
same edge (overall, the number of edges is as most doubled)
\cite{AKPW95,Madry10}.

In our distributed implementation of Madry's $j$-tree construction, we
adapt the parallel low average stretch spanning tree algorithm from
\cite{blelloch14} to our setting. The algorithm of \cite{blelloch14}
already works in a mostly decentralized fashion and we can therefore
also apply it in a distributed setting. In \sectionref{sec:spanning}, we show
how to run the algorithm of \cite{blelloch14} on a distributed cluster
graph. We note that the low average stretch spanning tree algorithm of
\cite{blelloch14} directly tolerates multi-edges as described above,
even if the same physical edge has to be used to represent multiple
edges between the same clusters.

\begin{theoremRF}{theorem:spanning}
  Suppose $\cG$ is a multigraph obtained from $G$ by assigning
  arbitrary edge lengths in $2^{n^{o(1)}}$ to the edges of $G$ (known
  to incident nodes) and performing an arbitrary sequence of
  contractions. Then we can compute a rooted spanning tree of $\cG$ of
  expected stretch $2^{O(\sqrt{\log n \log \log n})}$ within
  $(\sqrt{n}+D)n^{o(1)}$ rounds, where the edges of the tree in $\cG$ and their
  orientation is locally known to the endpoints of the corresponding
  edges in $G$.
\end{theoremRF}

Given the spanning tree $\cT=(\cV_\cT,\cE_\cT)$ of $\cG$, we need to compute
capacities $\Cp_{\cT}(e)$ for the spanning edges such that $\cG$ is embeddable
into $\cT$, which essentially boils down to computing a the absolute value of
the multicommodity flow $\vf'$ routing $\Cp(e)$ units of flow on $\cT$ for each
$e\in E$. As routing is trivial in trees, $\vf'$ is unique. Once $|\vf'|$ is
computed, the edge capacities of $\cT$ can be chosen accordingly and it is
straightforward to pick a suitable $\lambda_i$ and update the length function
$\ell(e)$. However, computing the absolute value of the multicommodity flow
\emph{fast} in the distributed setting requires some work.

\paragraph{Computing the Multicommodity Flow.}
In the following, assume that the edges of $\cT$ are oriented towards the root,
i.e., we will write $(c,\hat{c})\in\cT$ if $\hat c \in \cV$ is the parent of cluster $c\in
\cV$. Denote by $\cT_c$ the subtree of $\cT$ rooted at $c$. When embedding $\cG$
into $\cT$, we have to route a total of
\begin{equation*}
|\vf'(c,\hat{c})|=\sum_{\substack{c_1\in \cT_c\\ c_2\notin \cT_c}}
\sum_{\{c_1,c_2\}_{uv}\in \cE}\Cp(\{c_1,c_2\}_{uv})
\end{equation*}
commodity through the edge $\set{c,\hat{c}}\in \cT$, where we indexed the different
edges of the multigraph $(V,\cE)$ by using that for each $e\in \cE$ between
$c_1$ and $c_2$, $\psi$ maps $e$ to a unique edge $\{u,v\}\in E$ with $u\in c_1$
and $w\in c_2$. Note that $u$ and $v$ know that $\{c_1,c_2\}_{uv}\in \cE$, that
they are in the clusters $c_1$ and $c_2$, respectively, and what
$\Cp(\{c_1,c_2\}_{uv})$ is. We thus have to solve the task of determining this
sum for each edge $\set{c,\hat{c}}\in \cT$ via computations on the graph $G$
underlying $(\cV,\cE)$.

Observe that the spanning trees of the clusters together with $\cT$ induce a
(rooted) spanning tree $T$ of $G$. Essentially, we would like to perform, for
each edge $\set{c,\hat{c}}\in \cT$, an aggregation on $T$ and pipeline these
aggregations to achieve good time complexity. However, as shown in the following
lemma, the result is a running time linear in the depth of the tree, which may
be $\Omega(n)$ irrespective of $D$.
\begin{lemma}\label{lemma:naive_aggregation}
If $T$ has depth $d$, for each edge $e=(c,\hat{c})\in \cT$, $c$ can determine
$|\vf'(e)|$ within $O(d)$ rounds.
\end{lemma}
\begin{proof}
Consider the following algorithm.
\begin{compactenum}
\item For each cluster, all of its nodes learn the ancestors of the cluster in
$\cT$.
\item For each edge $\{c_1,c_2\}_{uv}\in \cE$, $u$ and $v$ exchange the
ancestor lists of $c_1$ and $c_2$.
\item Each node $u\in c_1\in \cT$ locally computes for each ancestor $c$ of
$c_1$ the value 
\begin{equation*}
\Cp_c(u):=\sum_{\substack{\{c_1,c_2\}_{uv}\in \cE\\
c\mbox{ is not ancestor of }c_2}}\Cp(\{c_1,c_2\}_{uv}).
\end{equation*}
\item For each edge $e=(c,\hat{c})\in \cT$, we aggregate
$\sum_{u\in T_c}\Cp_c(u)$ on $T_c$, where $T_c$ is the subtree of $T$
corresponding to $\cT_c$.
\end{compactenum}
Observe that, by definition,
\begin{equation*}
|\vf'(c,\hat{c})|=\sum_{\substack{c_1\in \cT_c\\ c_2\notin \cT_c}}
\sum_{\{c_1,c_2\}_{uv}\in \cE}\Cp(\{c_1,c_2\}_{uv})=\sum_{u\in T_c}\Cp_c(u),
\end{equation*}
as each edge $\{c_1,c_2\}_{uv}\in \cE$ with $c_1\in \cT_c$ and $c_2 \in
\cT\setminus \cT_c$ satisfies that either $u\in T_c$ or $v\in T_c$. Hence, it
remains to show that the above routine can be implemented with a running time of
$O(d)$.

Clearly, the first step takes $d$ rounds: $T$ has depth $d$, and we can perform
concurrent floodings on all subtrees without causing contention. The second step
requires at most $d-1$ rounds, as no node has more than $d-1$ ancestors. The
third step requires local computations only. Finally, the fourth step can be
performed in $d$ rounds as well, since we can perform concurrent convergecasts
on all subtrees without causing contention.
\end{proof}

To handle the general case, i.e., $d\gg \sqrt{n}$, we first decompose $\cT$ into
$O(\sqrt{n})$ parts of small diameter. There are different ways to achieve such
a decomposition of $\cT$ efficiently in a distributed way (e.g., by using
techniques from \cite{Kutten-Peleg}). The easiest way is to use randomization.
Suppose $c'$ is the parent cluster of non-root cluster $c$. We sample edge
$e=(c,\hat{c})\in \cT$ into the edge set $\cR$ with independent probability
$q_e:=\min\set{1,|c|/\sqrt{n}}$. Then, w.h.p.\ the forest $T\setminus \psi(\cR)$
consists $\tilde{O}(\sqrt{n})$ trees of depth $\tilde{O}(\sqrt{n})$.
\begin{lemma}\label{lemma:smallstrongdiameter}
  Let $\cT$ be a rooted spanning tree of a cluster (multi-)graph $\cG$ and let
  $\cR$ be a subset of the edges chosen at random as described above, and assume
  that the spanning tree of each cluster has depth at most $d$. W.h.p.,
  the forest $T\setminus \psi(\cR)$ induced by the edges $\cT\setminus\cR$ and
  the cluster spanning trees consists of $O(\sqrt{n})$ rooted trees of
  depth $d+O(\sqrt{n}\log n)$.
\end{lemma}
\begin{proof}
   Clearly, the number of trees in the forest induced by
   $\cE_\cT\setminus \cR$ is equal to $|\cR|+1$. The expected value
   for $|\cR|$ is given by the sum of the probabilities $q_e$ and thus
   $\E[|\cR|]\leq \sqrt{n}$. A standard Chernoff bound implies that 
   $|\cR|$ does not exceed $\sqrt{n}$ by more than a
   constant factor with high probability.
   
   To bound the depth of each (rooted) tree in $T\setminus \psi(\cR)$, consider
   a cluster $c\in \cT$ and a path $p$ from the leader $\{r\}:=\cL\cap c$ to
   some node in the subtree $T_r$ of $T$ rooted at $r$. The depth of $T_r\cap c$
   is bounded by $d$. Denote by $E_p$ the set of edges of $p$ that correspond to
   edges in $\cT$, i.e., each $e\in E_p$ satisfies that $e=\psi(e')$ for some
   $e'\in \cT$. Denote by $c(e)$ the child cluster of $e$, i.e., the endpoint
   further away from the root of $T$. By construction, $c(e)\in e'$ is also the
   cluster further away from the root of $\cT$ (for the $e'\in \cT$ with
   $\psi(e')=e$). Therefore, the length of $p$ is bounded by
   \begin{equation*}
   d+|E_p|+\sum_{e\in E_p} |c(e)|\,
   \end{equation*}
   i.e., the sum of the number of its edges in $c$, the number of edges between
   clusters $|E_p|$, and the number of edges in each traversed cluster. For
   $e\in E_p$, we have that $q_e=\min\{1,|c(e)|/\sqrt{n}\}$, yielding that
   \begin{equation*}
   \E[|\cR \cap \psi^{-1}(E_p)|]\leq \sum_{e\in E_p}\frac{|c(e)|}{\sqrt{n}}.
   \end{equation*}
   Applying Chernoff's bound shows that $\E[|\cR \cap \psi^{-1}(E_p)|]\in
   O(\log n)$ or $|\cR \cap \psi^{-1}(E_p)|>0$ w.h.p. The former implies that
   (i) $\sum_{e\in E_p}|c(e)|\in O(\sqrt{n}\log n)$ and (ii) $|E_p|\in
   O(\sqrt{n}\log n)$, as always $|c(e)|\geq 1$; it follows that the length of
   $p$ is bounded by $d+O(\sqrt{n}\log n)$ as claimed. The latter implies that,
   w.h.p., $p$ is not contained in $T\setminus \psi(\cR)$. As the number of
   simple paths in a tree is bounded by $O(n^2)$, applying the union bound
   completes the proof.
\end{proof}

For simplicity, in the following we assume that the high probability statements
of the above lemma hold with certainty; the final statements then follow by
applying the union bound.

Throughout the construction, we will maintain that edges of $\cR$ are never
retained. As there will be $o(\log n)$ levels of recursion, by inductive use of
the above lemma, it follows that clusters always have spanning trees of depth
$\tilde{O}(\sqrt{n})$. Exploiting this property together with the small number
of connected components of $\cT\setminus \cR$, we obtain a fast routine for the
general case.
\begin{lemma}\label{lemma:decomposition_aggregation}
Within $\tilde{O}(\sqrt{n}+D)$ rounds, for each edge $e=(c,\hat{c})\in \cT$, $c$ can
determine $|\vf'(e)|$.
\end{lemma}
\begin{proof}
Denote by $\cC$ the connected components of $\cT\setminus \cR$. For $c\in \cT$,
denote by $C\in \cC$ its connected component. We rewrite the (absolute value of)
the multicommodity flow
\begin{equation*}
|\vf'(c,\hat{c})|=
\sum_{\substack{c_1\in \cT_c\\ c_2\notin \cT_c}}
\sum_{\{c_1,c_2\}_{uv}\in \cE}\Cp(\{c_1,c_2\}_{uv})
=|\vf_1'(c,\hat{c})|+|\vf_2'(c,\hat{c})|-|\vf_3'(c,\hat{c})|,
\end{equation*}
where
\begin{align*}
|\vf_1'(c,\hat{c})|&:=
\sum_{\substack{c_1\in \cT_c\setminus C\\ c_2\notin \cT_c\setminus C}}
\sum_{\{c_1,c_2\}_{uv}\in \cE}\Cp(\{c_1,c_2\}_{uv}),\\
|\vf_2'(c,\hat{c})|&:=
\sum_{\substack{c_1\in \cT_c\cap C\\ c_2\notin\cT_c}}
\sum_{\{c_1,c_2\}_{uv}\in \cE}\Cp(\{c_1,c_2\}_{uv}),\mbox{ and}\\\\
|\vf_3'(c,\hat{c})|&:=
\sum_{\substack{c_1\in \cT_c\cap C\\ c_2\in \cT_c\setminus C}}
\sum_{\{c_1,c_2\}_{uv}\in \cE}\Cp(\{c_1,c_2\}_{uv}).
\end{align*}

Note that $|\vf_1'(c,\hat{c})|$ does only depend on the component of $c$, i.e.,
we need to determine and make known only $|\cC|\in O(\sqrt{n})$ values to cover
this term. For the other terms, we will reduce the problem to an aggregation
on the spanning tree of $C$ in the vain of \lemmaref{lemma:naive_aggregation}.

Concerning $|\vf_1'(c,\hat{c})|$, we employ the following routine.
\begin{compactenum}
\item Using its spanning tree, each component $C\in \cC$ determines a unique
identifier (say, the smallest cluster identifier) and makes it known to all its
nodes.
\item For each $\{c_1,c_2\}_{uv}\in \cE$, $u$ and $v$ exchange their component
identifiers.
\item The list of component identifiers and edges $(C,\hat{C})$ for each
$(c,\hat{c})\in \cR$ is made known to all nodes. This enables each node to
locally compute the tree resulting from contracing the components $C\in \cC$ in
$\cT$.
\item For each $C\in \cC$, fix an arbitrary $c\in C$. Each node $u\in
\cT_c\setminus C$ locally computes
\begin{equation*}
\Cp_C(u):=\sum_{\substack{\{c_1,c_2\}_{uv}\in \cE\\
c_2\notin \cT_c\setminus C}}\Cp(\{c_1,c_2\}_{uv});
\end{equation*}
we set $\Cp_c(u):=0$ for all $u\notin \cT_c\setminus C$ (nodes can determine
whether they are in $\cT_c\setminus C$ based on the information collected in the
previous two steps).
\item For each $C\in \cC$, make $\sum_{u\in V}\Cp_C(u)$ known to all nodes via
a BFS tree of $G$. For all $c\in C$, we have that
$|\vf_1'(c,\hat{c})|=\sum_{u\in V}\Cp_C(u)$.
\end{compactenum}
As discussed earlier, components' spanning trees have depth
$\tilde{O}(\sqrt{n})$ and $|\cC|\in O(\sqrt{n})$. Hence, Step 1 takes
$\tilde{O}(\sqrt{n})$ rounds and Steps 3 and 5 take $O(\sqrt{n}+D)$ rounds. Step
$2$ requires only one round of communication and Step 4 is local. Overall, the
routine requires $\tilde{O}(\sqrt{n}+D)$ rounds.

To determine $|\vf_2'(c,\hat{c})|$ and $|\vf_3'(c,\hat{c})|$ for each $c$, we
proceed similarly to \lemmaref{lemma:smallstrongdiameter}.
\begin{compactenum}
\item For each $C\in \cC$ and each $c\in C$, all nodes in $c$ learn the list of
ancestors of $c$ that are in $C$ (using the spanning tree of $C$ in $G$).
\item For each $\{c_1,c_2\}_{uv}\in \cE$, $u$ and $v$ exchange their
component identifiers, as well as the ancestor lists determined in the previous
step.
\item The list of component identifiers and edges $(C,\hat{C})$ for each
$(c,\hat{c})\in \cR$ is made known to all nodes.
\item For each $C\in \cC$, $c\in C$, and $u\in \cT_c\cap C$, $u$ locally
computes
\begin{equation*}
\Cp_c(u):= \sum_{\substack{\{c_1,c_2\}_{uv}\in \cE\\
v\in c_2\notin \cT_c}}\Cp(\{c_1,c_2\}_{uv})
-\sum_{\substack{\{c_1,c_2\}_{uv}\in \cE\\
v\in c_2\in \cT_c\setminus C}}\Cp(\{c_1,c_2\}_{uv}).\footnote{If $v\in C$, $u$ can
decide whether $v\in \cT_c$ based on the ancestor lists. If $v\notin C$, $u$
can decide whether $v\in \cT_c$ based on $v$'s component identifier and the
information collected in Step 3.}
\end{equation*}
\item For each edge $e=(c,\hat{c})\in \cT$, we aggregate
$\sum_{u\in T_c\cap C}\Cp_c(u)$ on $T_c\cap C$, where $T_c$ is the subtree
of $T$ (the spanning tree of $G$) corresponding to $\cT_c$.
\end{compactenum}
Note that, by definition of $\Cp_c(u)$, we have that
\begin{equation*}
\sum_{u\in T_c\cap
C}\Cp_c(u)=|\vf_2'(c,\hat{c})|-|\vf_3'(c,\hat{c})|.
\end{equation*}
Hence, it remains to analyze the running time of this second subroutine. Again,
using that components' spanning trees have depth $\tilde{O}(\sqrt{n})$ and that
$|\cC|\in O(\sqrt{n})$, we can conclude that Steps 1, 2, and 5 take
$\tilde{O}(\sqrt{n})$ rounds, while Step 3 takes $O(\sqrt{n}+D)$ rounds. As Step
4 requires local computation only, the resulting running time is
$\tilde{O}(\sqrt{n})$ rounds. Overall, we conclude that $|\vf'(c,\hat{c})|$ can
be computed for each $c$ within $\tilde{O}(\sqrt{n}+D)$ rounds, by running each
of the two subroutines and summing up their outputs.
\end{proof}

\subsection{\texorpdfstring{Approximating \boldmath$\cG$ by a Distribution
over Simpler Graphs}{Approximating by a Simpler Distribution}}
\label{sec:computedistribution}

Using the techniques of \cite{raecke08} and the above construction of
low average stretch spanning trees, it is possible to design a
distributed algorithm to compute a distribution of such spanning trees
which approximates the cut structure of the underlying network graph
within a factor $2^{O(\sqrt{\log n \log\log n})}$ (i.e., in the order of the
average stretch of the computed spanning trees). However, when doing
this, the number of spanning trees we need to compute can be linear in
the size of $\cG$. We follow the same general idea as Sherman, wo applied the
construction by Madry~\cite{Madry10} recursively to decrease the \emph{step}
complexity, to avoid this sequential bottleneck and achieve a small \emph{time}
complexity in the distributed setting.

For each edge $e\in \cT$, we define $\rload_{\cT}(e):=\Cp_{\cT}(e)/\Cp(e)\geq 1$
to be the relative load of $e$ (edges $e\in\cE\setminus \cT$ have
$\rload_{\cT}(e)=0$). The construction of \cite{raecke08} builds up a potential
for each edge $e$ of $\cG$, where with each new tree added to the distribution,
the potential of $e$ grows by a term proportional to
$\rload_{\cT}(e)/\max_{e'\in \cE} \{\rload_{\cT}(e')\}$. The potential of each
edge is bounded by $\alpha=n^{o(1)}$ and hence with every additional spanning
tree, we are guaranteed to make progress for all edges $e\in \cE$ with
$\rload_{\cT}(e)$ close to $\max_{e'\in \cE} \{\rload_{\cT}(e')\}$. In the worst
case, this can just be a single edge for each spanning tree $\cT$. The key idea
of Madry \cite{Madry10} is to augment the tree $\cT$ with additional edges in
order to reduce the maximum relative load so that in the new graph, a large
number of edges have a relative load close to the maximum one.

Basically, we can reach a large number of edges with relative load
close to the maximum relative load by repeatedly deleting the edge
with largest relative load until a large number of the remaining edges
has a relative load that is within a constant factor of the remaining
maximum relative load. When deleting some edges of $\cT$, one has to
add back some of the original edges of $\cG$ in order to maintain the
property that $\cG$ is embeddable into the resulting graph. Formally, let
$\cF\subseteq \cT$ be a subset of the spanning tree edges. The
edge set $\cT\setminus \cF$ defines a spanning forest of $\cG$
consisting of $|\cF|+1$ components. We define a subgraph
$\cH(\cT,\cF)$ of $\cG$ as follows. The node set of $\cH(\cT,\cF)$ is
$\cV$. Further, $\cH(\cT,\cF)$ contains all edges in $\cT\setminus
\cF$ and it contains all edges $\set{c,c'}_{uv}\in \cE$ of $\cG$ for which
$c$ and $c'$ are in different components in the forest induced by the
edges in $\cT\setminus \cF$. Let $\cE_\cH$ be the set of edges of
$\cH$. We set the capacities $\Cp_{\cH}(e)$ of edges $e\in \cE_\cH$ to
be $\Cp_{\cH}(e):=\Cp_{\cT}(e)$ if $e\in \cT\setminus \cF$ and
$\Cp_{\cH}(e):=\Cp(e)$ otherwise. Note that this guarantees that $\cG$ is
$1$-embeddable into $\cH$. For the following discussion, we define
$\mathbb{H}[j]$ to be the set of graphs $\cH(\cT,\cF)$ for a spanning
tree $\cT$ of $\cG$ and a set of edges $\cF$ of $\cT$ of size
$|\cF|\leq j$.

Assume that the weighted average stretch of the spanning tree $\cT$ as given by
Eq.~\eqref{eq:weightedavgstretch} is upper bounded by $\alpha$. Also recall that
we assume that all capacities of $G$ are integers that are polynomially bounded
in the number of nodes $n$. As throughout our construction, each edge capacity
always approximately corresponds to the capacity of some cut in $G$, it is not
hard to guarantee that all capacities of $\cG$ are integers between $1$ and
$\poly(n)$. Given a spanning tree $\cT$ of $\cG$, let $R:=\max_{e\in \cT}
\{\rload_\cT(e)\}$ be the largest relative load of all edges of $\cT$. In order
to determine the set of edges $\cF$, we start by partitioning the edges in $\cT$
into $i_{\max}= O(\log n)$ classes $\cF_1,\dots,\cF_{i_{\max}}$, where class
$\cF_i$ contains all edges with relative load in $(R/2^i,R/2^{i-1}]$. Now, for
any $j_0\leq |\cT|$, there exists an edge class $\cF_i$ such that
$\big|\bigcup_{i'<i} \cF_{i'}\big|\leq j_0$ and $|\cF_i| \geq j_0/i_{\max}
=\Omega(j_0/\log n)$; otherwise, $|\cT|=\big|\bigcup_i \cF_i\big|< j_0\leq
|\cT|$, a contradiction. We define $\cF':=\bigcup_{i'<i} \cF_{i'}$.

In \cite{Madry10arxiv}, this set of edges is used to construct the graph
$\cH(\cT,\cF')$. For the distributed computation, it will be useful to have a
graph $\cH(\cT,\cF)$ in which all the trees of the forest induced by
$\cE_\cT\setminus \cF$ have small diameter. We therefore rely on the same
technique as for computing the capacities of $\cT$ and remove a few random
additional edges of $\cT$. In fact, we can simply use the same subset of edges
$\cR\subseteq \cT$ that has been determined and used before, prior to
\lemmaref{lemma:smallstrongdiameter}. We define $\cF:=\cF' \cup \cR$ and use the
graph $\cH=(\cV,\cE_\cH,\Cp_\cH):=\cH(\cT,\cF)$. Since all the edges of $\cT$
with $\rload_\cT(E)>R/2^{i-1}$ are removed, all edges of $\cH$ have relative
load at most $R/2^{i-1}$. Further, all the $\Omega(j/\log n)$ edges of $\cF_i$
have relative load larger than $R/2^i$. Based on Theorem 5.2 and Corollary 5.6
of \cite{Madry10arxiv} and on \theoremref{theorem:spanning}, we can show the
following lemma.

\todo{I believe that we need to also apply
\lemmaref{lemma:decomposition_aggregation} here; apparently it's not
officially used anywhere so far.}
\begin{lemma}\label{lemma:Hdistribution}
  Given are a distributed cluster (multi-)graph $\cG=(\cV,\cE,\Cp_\cG)$
  consisting of $|\cV|=N$ clusters and $|\cE|=N\polylog(n)$ edges and
  a parameter $j\geq 1$ such that that $j=\omega(\sqrt{n}\log
  n)$. There is a distributed algorithm to compute an $\big(2^{O(\sqrt{\log N
      \log\log N})},\mathbb{H}[j]\big)$-distribution of $\cG$ on $
  2^{O(\sqrt{\log N \log\log N})}\cdot N/j$ graphs, which runs in the \Congest
  model on the underlying network graph $G$ in $(D+\sqrt{n})\cdot n^{o(1)}\cdot
  N/j$ rounds.
\end{lemma}
\begin{proof}
  Let $\alpha=2^{O(\sqrt{\log n \log\log n})}$ be the average stretch
  guarantee of the spanning tree algorithm. It follows directly from
  Theorem 5.2 and Lemma 5.5 in \cite{Madry10arxiv} that we can compute
  an $(O(\alpha),\mathbb{H}[|\cE|\cdot \alpha\log(n)/s])$-decomposition of $\cG$
  on $s$ graphs in time $s\cdot T_{\mathrm{tree}}$ if the following conditions
  are satisfied:
  \begin{enumerate}[(1)]
  \item The time for computing one low average stretch spanning tree is
    upper bounded by $T_{\mathrm{tree}}$.
  \item Given $\cT$ and the set of edges $\cF$ as computed above, let
    $\rload_{\max}:=\max_{e\in \cT\setminus
      \cF}\{\rload_{\cT}(e)\}$. The number of edges of $\cH(\cT,\cF)$ with
    relative load at least $\rload_{\max}/2$ is $\Omega(|\cE|\alpha/s)$.
  \end{enumerate}
  As observed above, in $\cT$, all edges in the set $\cF_i$ have
  relative load between $\rload_{\max}/2$ and $\rload_{\max}$. When
  constructing $\cH(\cT,\cF)$, the relative load of edges in
  $\cT\setminus \cF$ does not change and thus, all nodes in
  $\cF_i\setminus \cF = \cF_i\setminus \cR$ have relative load between
  $\rload_{\max}/2$ and $\rload_{\max}$. Recall that $|\cF_i| =
  \Omega(j/\log n)$. By \lemmaref{lemma:smallstrongdiameter}, with
  high probability, we have $|\cR|=O(\sqrt{n})$.  Since we assumed
  that $j=\omega(\sqrt{n}\log n)$, we have $|\cF_i| = \omega(|\cR|)$
  and thus $|\cF|=\Omega(j/\log n)$. The second condition is now
  satisfied by choosing $s=\Theta(|\cE|\alpha\log(n)/j)=2^{O(\sqrt{\log N
      \log\log N})}\cdot N/j$. 

  Assuming that the time to compute a single low average stretch
  spanning tree can be upper bounded by
  $T_{\mathrm{tree}}=(D+\sqrt{n})\cdot n^{o(1)}$, the lemma now
  follows. By \theoremref{theorem:spanning}, this is guaranteed as
  long as all edge lengths are integers between $1$ and
  $2^{n^{o(1)}}$. Inspecting the construction in \cite{Madry10arxiv}
  and \cite{raecke08}, we can observe that the edge lengths cannot get
  larger than a value exponential $\alpha$. By rounding them to integers, we
  introduce an additional multiplicative error of factor $2$, which does not
  affect the asymptotic behavior. As $\alpha=2^{O(\sqrt{\log n\log\log n})}$,
  the claim of the lemma follows.
\end{proof}

\subsection{Transforming \texorpdfstring{\boldmath$\cH(\cT,\cF)$ into a
$j$-Tree}{the Simpler Family into j-Trees}}
\label{sec:TransTojTrees}

Given a graph $\cH(\cT,\cF)\in\mathbb{H}[j]$, it remains to transform
$\cH(\cT,\cF)$ into an $O(j)$-tree $\cJ$ such that the two graphs are
$O(1)$-embeddable into each other. In the following, we first describe
the construction and we formally prove that the resulting $O(j)$-tree
$\cJ$ and the given graph $\cH(\cT,\cF)$ are $O(1)$-embeddable into
each other. We then show how to efficiently construct $\cJ$ in a
distributed way.

Assume that we are given a spanning tree $\cT$ of $\cG$ and a graph
$\cH(\cT,\cF)$ which is constructed as described above. Consider the
forest induced by the edges in $\cT\setminus \cF$.


Let $P_1\subseteq \cV$ be the set of clusters of $\cT\setminus \cF$ which are
incident to one of the deleted tree edges in $\cF$. We call $P_1$ the
\emph{primary portals} of $\cT\setminus \cF$. Given $P_1$, we define the
skeleton $\cS_{\cT\setminus \cF}$ of $\cT\setminus \cF$ as follows.
$\cS_{\cT\setminus \cF}$ is obtained from $\cT$ by repeatedly deleting
non-portal clusters of degree $1$ until all remaining clusters are either in
$P_1$ or they have degree at least $2$. Denote by $P_2$ all clusters of degree
larger than $2$ that are not primary portals; $P_2$ are the \emph{secondary
portals} portals. The set of all portal clusters now is $P:=P_1\cup P_2$.
The skeleton $\cS_{\cT\setminus \cF}$ is thus a forest consisting of a set of
portals and paths connecting them, where all inner clusters of these paths
have degree $2$.

Given the skeleton $\cS_{\cT\setminus \cF}$, consider one of these paths $\cP$.
In the last step, we remove the edge with the smallest capacity from each
such $\cP$. In doing so, we split the forest into trees so that each tree
contains exactly one portal. It is straightforward to bound the number of
resulting trees in terms of $\cF$.

\begin{lemma}\label{lemma:nofPortals}
  Let $\cH(\cT,\cF)\in \mathbb{H}[j]$, i.e., $|\cF|\leq j$. Then, in the above
  construction, the total number of portal nodes is less than $4j$.
\end{lemma}
\begin{proof}
  Clearly, $|P_1|\leq 2|\cF|\leq 2j$. As when computing the skeleton
  $\cS_{\cT\setminus \cF}$, non-portal clusters of degree $1$ are successively
  removed, we obtain a forest whose leaves are primary portals. As the sum of
  the degrees in an $N$-node forest is at most $2(N-1)$, the number of nodes of
  degree at least $3$ is upper bounded by the number of leaves minus $2$. We
  conclude that $|P_2|<|P_1|\leq 2j$, and hence $|P|<4j$.
\end{proof}

Finally, we identify each of the resulting trees with its portal and logically
move all edges between different trees to the portals. For each edge
$e=\set{c,c'}\in \cE_\cH\setminus (\cT\setminus\cF)$ (i.e., each non-tree edge
of $\cH(\cT,\cF)$), we add a virtual edge of capacity $\Cp_{\cG}(e)$ between the
portals of the trees containing $c$ and $c'$, respectively. Further, let $\cD$
be the set of edges that were deleted from the paths of degree-$2$ clusters
connecting portals in the skeleton. For every edge $e\in \cD$, we add a virtual
edge of capacity $\Cp_{\cH}(e)=\Cp_{\cT}(e)$ between the two portals that were
connected by the path from which $e$ was deleted.

Let us summarize this part of the construction; see \figureref{fig:Clusters}
for an example of a possible result. Starting from a forest $\cT\setminus
\cF$, do as follows:
\begin{compactenum}
\item Define $P_1$ as the endpoints of edges in $\cF$;
\item iteratively delete degree-$1$ clusters that are not in $P_1$ until this
process halts;
\item define $P_2$ as the clusters retaining degree larger than $2$ that are not
in $P_1$ and set $P:=P_1\cup P_2$;
\item delete from each (maximal) path without clusters from $P$ the edge $e\in
\cD$ of minimum capacity and replace it by an edge of the same capacity between
its endpoints; and
\item for each edge $e\in \cE_{\cH}$ between different components of
$\cT\setminus (\cF\cup \cD)$, add an edge of the same capacity between the
unique portals in these components.
\end{compactenum}
Hence, the resulting graph consists of the forest induced by
$\cT\setminus(\cF\setminus \cD)$ and (possibly parallel) edges between the
unique portals of the trees of the forest. By \lemmaref{lemma:nofPortals}, the
number of such portals is smaller than $4j$, implying that the resulting graph
is a $4j$-tree. In the following, we denote this $4j$-tree by $\cJ$.

\begin{figure}[t]
    \centering
    \subfloat[]{{\includegraphics[width=7cm]{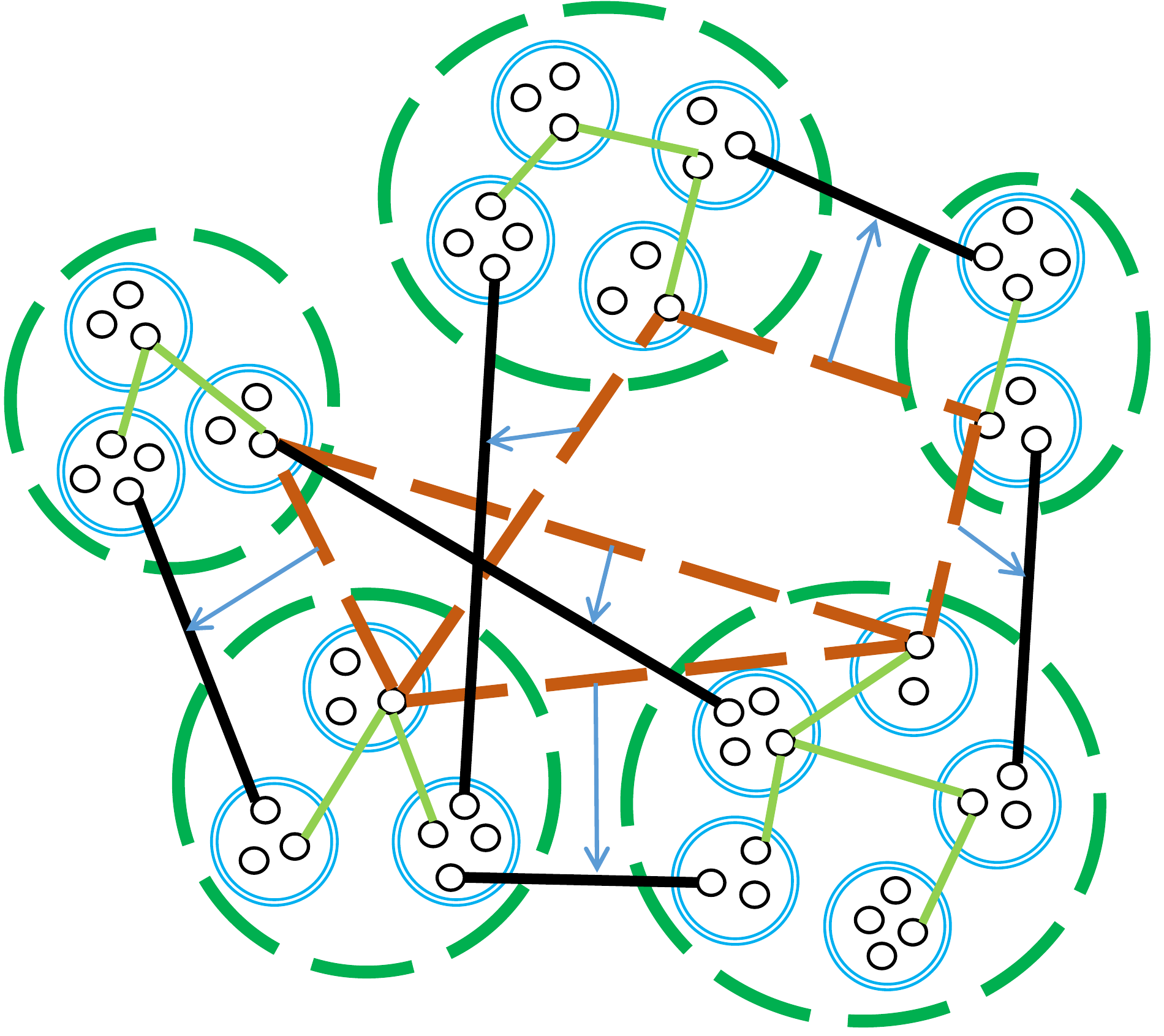} }}%
    \qquad
    \subfloat[]{{\includegraphics[width=7cm]{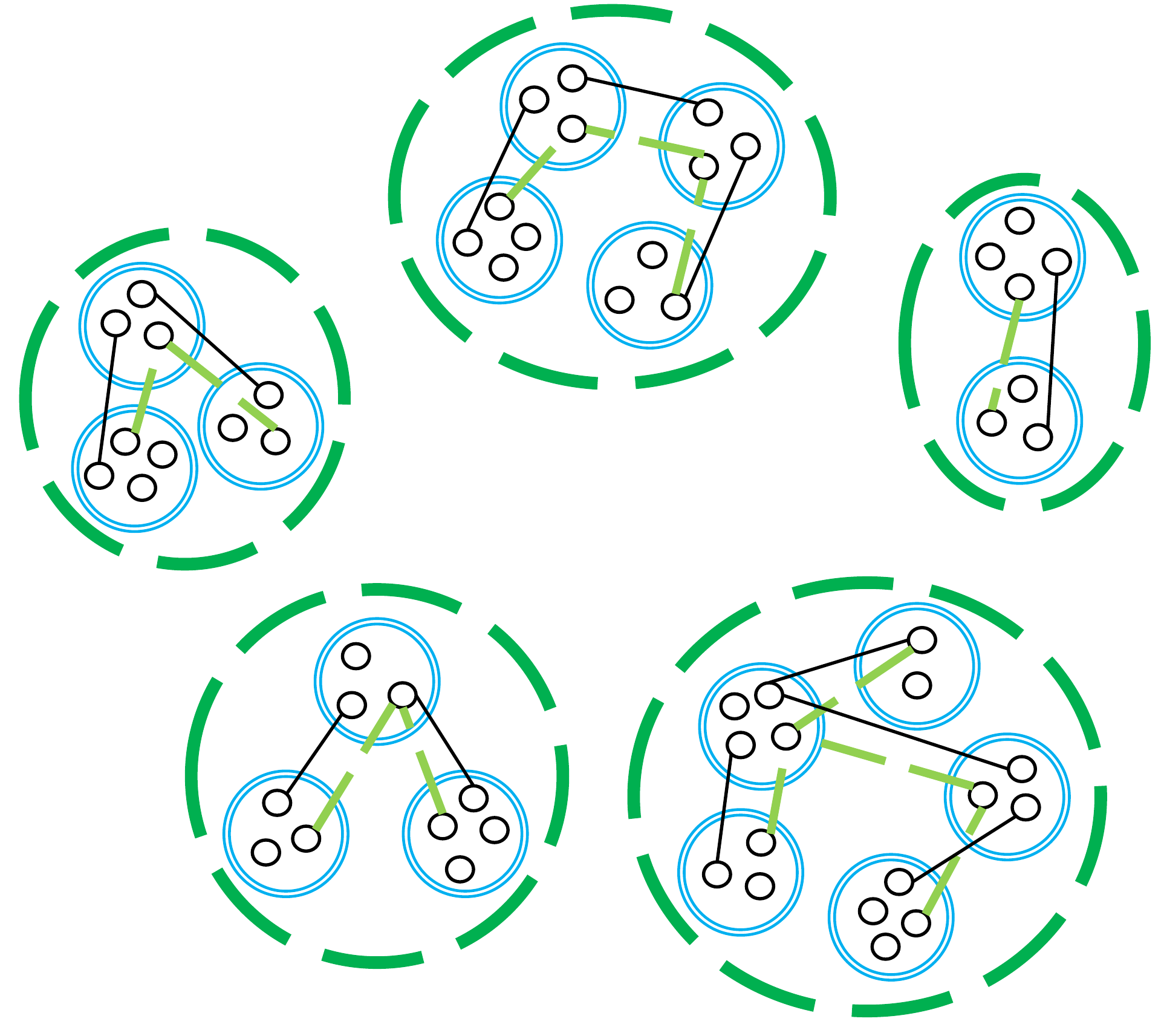} }}%
    \caption{\small An example $j$-tree at the second level of
    recursion.
    On the left side,
    green circles indicate the components of the forest of this $j$-tree, which
    are each made of a number of $1$-clusters, indicated by blue
    double-line circular shapes. Edges inside $1$-clusters are not
    shown. Solid straight green edges indicate virtual edges
    of level-$1$ that became edges of the level-$2$ forest. For each
    of these green edges, there is a real edge between some two nodes of the
    same two level-$1$ components (not shown). Brown edges represent
    (virtual) core edges, and they are
    mapped (blue arrows) to corresponding real edges (solid black). On
    the right side, real edges related to the virtual forest edges are
    represented as solid black edges between two nodes of the same connected
    component.}
	\label{fig:Clusters}
\end{figure}

\paragraph{Mutual Embeddability of $\cH(\cT,\cF)$ and $\cJ$.}
Before discussing how to efficiently construct (and represent) $\cJ$ in a
distributed way, let us first show that $\cH(\cT,\cF)$ and $\cJ$ are
$O(1)$-embeddable into each other.

The proofs of the following two lemmas is very similar to the corresponding
result by Madry~\cite{Madry10}. However, since our $j$-tree construction
slightly deviates from Madry's, the claims do not readily follow from any
lemma in \cite{Madry10,Madry10arxiv}.
\begin{lemma}\label{lemma:H_into_j}
$\cH(\cT,\cF)$ is $O(1)$-embeddable into $\cJ$.
\end{lemma}
\begin{proof}
There are three types of edges of $\cH(\cT,\cF)$ to distinguish: 
  \begin{compactenum}[a)]
  \item edges in $(\cT\setminus \cF)\setminus \cD$,
  \item edges in $\cD$, and
  \item the remaining edges from $\cE_{\cH}$ connecting different trees of the
  forest $\cT\setminus \cF$.
  \end{compactenum}
  Case a) is the most straightforward, as all these edges are
  also present in $\cJ$ with the same capacity. Edges from $e\in \cD$
  were deleted from a path $\cP$ connecting two portals in the skeleton. In
  $\cJ$, they can therefore be routed through the path $\cP$ and the virtual
  edge with capacity $\Cp_{\cT}(e)$ connecting the portal nodes at the ends of
  $\cP$. Because $e$ is the lowest capacity edge of $\cP$, this adds
  relative load at most $1$ to each edge of $\cP$.
  
  Finally, let us consider one of the remaining edges $e\in\cE_{\cH}$. The edge
  $e=\set{c_1,c_2}$ connects two trees $T_1\neq T_2$ of $\cJ$. Let us assume
  that $c_1\in T_1$ and $c_2\in T_2$. When routing from $c_1$ to $c_2$, we follow 
  \begin{compactenum}
  \item the path from $c_1$ in $T_1$ to the first skeleton cluster $s_1\in T_1$
  on the path to the (unique) portal $p_1\in T_1\cap P$,
  \item the skeleton path from $s_1$ to $p_1$,
  \item the virtual edge corresponding to $e$ between $p_1$ and the (unique)
  portal $p_2\in T_2\cap P$,
  \item the skeleton path from $p_2$ to the last skeleton cluster $s_2\in T_2$
  when going from $p_2$ to $c_2$ in $T_2$, and
  \item the path from $s_2$ to $c_2$ in $T_2$.
  \end{compactenum}
  Let us compare the path from $c_1$ via $s_1$ to $p_1$ with
  the path on which $e$ is routed on the spanning tree $\cT$. The part from
  $c_1$ to $s_1$ is also used when routing in $\cT$. If from $s_1$ we follow the
  same direction to $p_1$ as in $\cT$, the two paths are, in fact, identical
  up the point when we reach $p_1$. In this case, the capacities on this path
  suffice by construction, as we defined $\Cp_{\cT}(e')=|\vf'(e')|$. Let us
  therefore consider the case in which we set out in the opposite direction
  from $s_1$ on the skeleton path $\cP\ni s_1$ connecting $p_1$ to some other
  portal $p_2$ than we would in $\cT$. In that case,
  routing on $\cT$ would cross the edge $e'\in\cD$ that was deleted from $\cP$.
  Because $e'$ is the edge from $\cP$ of smallest capacity, the
  contribution to the relative load of all edges crossed on $\cP$ is upper
  bounded by the relative load contributed to $e'$ when routing on $\cT$. Again,
  summing over all edges from $\cE_{\cH}$ falling under Case c), this may
  increase their total relative loads only by an additive $1$. Trivially, the
  third step causes relative load $1$ on the virtual edge corresponding to $e$,
  since it is not used for routing any other edge. Reasoning symmetrically for
  Steps 4 and 5, we can conclude that embedding $\cH(\cT,\cF)$ into $\cJ$ leads
  to constant relative load on all edges.
\end{proof}

\begin{lemma}\label{lemma:j_into_H}
$\cJ$ is $O(1)$-embeddable into $\cH(\cT,\cF)$.
\end{lemma}
\begin{proof}
 All the edges of the trees of $\cJ$ are also present
  in $\cH(\cT,\cF)$ with the same capacity; they are hence
  straightforward to embed. Let us therefore consider the virtual
  edges connecting the portals of $\cJ$. There are two types of
  virtual edges, the ones representing edges from $\cD$ and those
  that correspond to the non-tree edges of $\cH(\cT,\cF)$. We first have a look
  at a virtual edge $e$ corresponding to an edge $e'\in \cD$. The edge $e$
  is routed by following path from which $e'$ was removed. Since
  the capacity of each edge on the path is at least the capacity
  of $e$, this contributes at most $1$ to the relative load of each
  edge.
  
  Now consider a virtual edge $e=\set{c,c'}$ corresponding to an edge $e'\in
  \cE_{\cH}$ of $\cH(\cF,\cT)$. The edge $e$ is routed on the trees of $\cJ$ the
  clusters $c$ and $c'$ reside in and via $e'$. The latter causes relative load
  $1$, as $e$ and $e'$ have the same capacity and no other edge uses $e'$.
  Similarly to the embedding of $\cH(\cT,\cF)$ into $\cJ$, the tree parts of the
  routing path that are subpaths of the path between $c$ and $c'$ in $\cT$
  and thus will not cause more than additive relative load $1$ when summing
  over all edges of this type to embed. If we diverge from this path, this is
  because an edge from $\cD$ lies on the routing path in $\cT$; analogously to
  \lemmaref{lemma:H_into_j}, following the skeleton path from which it was
  deleted to the respective portal increases the maximum relative load by at
  most an additional $1$.
\end{proof}

\paragraph{Distributed Implementation.}
Let us now move to the distributed implementation of the above
$4j$-tree construction. Recall that because $\cF$ includes the random
set of edges $\cR$, by \lemmaref{lemma:smallstrongdiameter}, all trees in
$\cT\setminus \cF$ have depth $\tilde{O}(\sqrt{n})$. With this in mind,
constructing the skeleton is fairly simple.

\begin{lemma}\label{lemma:skeleton}
  Given are a spanning tree $\cT$ of a distributed cluster graph and the
  set of tree edges $\cF$ as computed above. We can determine the skeleton
  $\cS_{\cT\setminus\cF}$, the set of portals $P$, and the set of edges $\cD$
  (i.e., for $\{c,c'\}_{uv}\in \cD$, $u$ and $v$ will learn this) in
  time $O(\sqrt{n}\log n)$ in the \Congest model on the underlying network
  graph. In the same time, we can also orient the trees rooted at the portals.
\end{lemma}
\begin{proof}
W.l.o.g., consider a single tree $T$ of the forest $\cT\setminus \cF$. By
\lemmaref{lemma:smallstrongdiameter}, the induced tree in $G$ has depth
$\tilde{O}(\sqrt{n})$. Perform the following steps:
\begin{compactitem}
\item For each edge $e\in \cF$, its incident clusters
learn\footnote{A cluster for which edges to children are in $\cF$ may not
``know'' about its incident edges in $\cF$ as a whole, but determining whether
there is at least one is trivial.} that they are primary portals, i.e., are in
$P_1$.
\item Iteratively mark non-portal clusters with at most $1$ marked neighboring
cluster until this process stops. Unmarked clusters are in the skeleton.
\item Unmarked clusters with more than two unmarked neighboring portals are
secondary portals.
\item The skeleton paths connecting portals find a minimum capacity edge and
add it to $\cD$.
\item Each tree of $\cT\setminus(\cF\cup \cD)$ is rooted at its unique
portal, whose identifier is made known to all nodes in the induced tree in $G$
(together with clusters' spanning trees).
\item These identifiers are exchanged with all neighbors in $G$.
\end{compactitem}
From the gathered information, for each edge $\{c,c'\}_{uv}\in \cJ$, $u$ and $v$
now can determine its membership and its capacity in $\cJ$. Observe that the
bound of $\tilde{O}(\sqrt{n})$ on the depth of the spanning trees of $G$
leveraged for communication in the above construction implies that all the above
steps can be completed in $\tilde{O}(\sqrt{n})$ rounds, which completes the
proof.
\end{proof}
The trees rooted at the portals now induce the clusters of the new cluster
graph.
\begin{corollary}\label{cor:jtree}
  Given a graph $\cH(\cT,\cF)\in\mathbb{H}[j/4]$ as computed above on a cluster
  graph whose clusters' spanning trees have maximum depth $d$, there is an
  $\tilde{O}(D+d+\sqrt{n})$-round distributed algorithm to compute
  \begin{compactitem}
  \item a cluster graph whose clusters' spanning trees have depth
  $d+\tilde{O}(\sqrt{n})$; and
  \item a $j$-tree $\cJ$ on this cluster graph, i.e., for each edge $e\in \cJ$,
  there is a corresponding graph edge $\{u,v\}\in E$ whose constituent nodes
  know that $e\in \cJ$ as well as $\Cp_{\cJ}(e)$; such that
  \item $\cH(\cT,\cF)$ is $1$-embeddable into $\cJ$ and $\cJ$ is
  $O(1)$-embeddable into $\cH(\cT,\cF)$; and
  \item the new clusters are induced by the tree components of $\cJ$.
  \end{compactitem}
\end{corollary}
\begin{proof}
This readily follows from
Lemmas~\ref{lemma:smallstrongdiameter}, \ref{lemma:nofPortals},
\ref{lemma:H_into_j}, \ref{lemma:j_into_H}, and~\ref{lemma:skeleton}. The only
thing left to note is that clusters can learn the number of nodes they contain
by a simple converge- and broadcast operation on their spanning trees.
\end{proof}

\subsection{Sampling from the Recursively Constructed Distribution}

We have now all pieces in place to efficiently sample from a distribution
similar to Sherman's in a distributed fashion. The difference is that
\theoremref{theorem:spanning} and thus \lemmaref{lemma:Hdistribution} merely
give $\alpha\in 2^{O(\sqrt{\log n \log \log n})}$, implying that we must use
fewer levels of recursion to ensure that the final approximation guarantee of
the congestion approximator will remain in $n^{o(1)}$.
\begin{theorem}\label{theorem:distribution}
W.h.p., within $(\sqrt{n}+D)n^{o(1)}$ rounds of the \Congest model, we
can sample a tree $\cT$ from a distribution of $n^{1+o(1)}$ (virtual)
rooted spanning trees on $G$ with the following properties.
\begin{compactitem}
\item For any cut of $G$ of capacity $C$, the capacity of the cut in $\cT$ is
at least $C$.
\item For any cut of $G$ of capacity $C$, the expected capacity of the cut in
$\cT$ is at most $\alpha C$, where $\alpha\in n^{o(1)}$.
\item The distributed representation of $\cT$ is given by a hierarchy
  of cluster graphs $\cG_i=(\cV_i,\cE_i,\cL_i,\mathfrak{T}_i,\psi_i)$,
  $i\in \{0,\ldots,i_0\}$, $i_0\in o(\log n)$, on network graph $G$,
  with the following properties.
\begin{compactitem}
\item The spanning trees of the clusters of $\cG_i$ have depth
$\tilde{O}(\sqrt{n})$.
\item $|\cV_{i_0}|=n^{1/2+o(1)}$.
\item $\cG_i$ is the (rooted) tree resulting from $\cT$ by contracting the
clusters of $\cG_i$.
\item For $i>0$, $\cG_i$ is also a cluster graph on network graph $\cG_{i-1}$.
\item For $i>0$, each cluster $c_i\in \cV_i$ of $\cG_i$, interpreted as cluster
graph on $\cG_{i-1}$, contains a unique \emph{portal cluster} $p(c_i)\in
\cV_{i-1}$ of $\cG_{i-1}$ that is incident\footnote{Note that the corresponding
physical edges in $G$ may still connect to different sub-clusters of $c_i$.} to
all edges of $\cG_i$ containing $c_i$. That is, $\cG_{i-1}$ is a $|\cV_i|$-tree
with core $p(\cV_i)$.
\end{compactitem}
\end{compactitem}
\end{theorem}
\begin{proof}
In the following, we will use w.h.p.\ statements as if they were deterministic;
the result then follows by taking the union bound over all (polynomially many
in $n$) such statements we use.

Set $\beta:=2^{\log^{3/4} n}$. To start the
recursion, we will use $G$ as cluster graph of itself. Formally,
$\tilde{\cG}_0:=(V,E,V,\{(\{v\},\emptyset)\}_{v\in V},\operatorname{id})$, where
$\operatorname{id}$ is the identity function. We perform the following
construction until it terminates:
\begin{compactenum}
\item Sparsify $\tilde{\cG}_{i-1}$ using \lemmaref{lemma:sparsification} for
some fixed constant $\varepsilon$, e.g., $\varepsilon = 1/2$. This takes
$(\sqrt{n}+D)n^{o(1)}$ rounds. Multiply all edge capacities by
$1/(1-\varepsilon)$ (so $\tilde{\cG}_{i-1}$ can be $1$-embedded into the sparser
graph).
\item If $|\cV_{i-1}|\notin \omega(\sqrt{n}\beta/\log n)$, set $i_0:=i$ and
stop. This takes $O(D)$ rounds by communicating over a BFS tree of $G$.
\item Apply \lemmaref{lemma:Hdistribution} for $j=|\cV_{i-1}|/(4\beta)$ to the
sparsified cluster graph; by the previous step, this choice of $j$ is feasible.
As $|\cV_{i-1}|/j=4\beta\in n^{o(1)}$, constructing the distribution requires
$(\sqrt{n}+D)n^{o(1)}$ rounds in total.
\item Sample a cluster graph from the distribution. This is done in $O(D)$
rounds letting some node broadcast $O(\log n)$ random bits over a BFS tree.
\item Apply \corollaryref{cor:jtree} to extract a $|\cV_{i-1}|/\beta$-tree 
of $\cG_{i-1}$. The corollary also yields a
cluster graph $\tilde{\cG}_i$ (which is also a cluster graph on network graph
$\cG_{i-1}$) so that each of its clusters $c_i$ contains
exactly one portal cluster $p(c_i)$ of the $|\cV_{i-1}|/\beta$-tree on
$\cG_{i-1}$. This step completes in $\tilde{O}(\sqrt{n}+D)$ rounds:
there are fewer than $\log_{\beta} \sqrt{n}\ll \log n$ iterations of the
overall construction, as $|\cV_i|\leq |\cV_{i-1}|/\beta$, implying that $d\in
\tilde{O}(\sqrt{n})$ for each application of \corollaryref{cor:jtree}.
\item Recurse on $\tilde{\cG}_i$, i.e., set $i:=i+1$ and go back to Step 1.
\end{compactenum}
When the above construction halts, we have that
$|\cV_{i_0-1}|=O(\sqrt{n}\beta)=n^{1/2+o(1)}$. Thus, we can make the
(sparsified) cluster graph $|\cG_{i_0-1}|$ known to all nodes in
$(\sqrt{n}+D)n^{o(1)}$ rounds via a BFS tree of $G$. We then continue the
construction locally without controlling the size of components, which removes
the constraint on $j$ when applying \lemmaref{lemma:skeleton}, until the core
becomes empty, i.e., we construct a tree.\footnote{This is essentially Sherman's
construction on the small constructed cluster graph.} We collapse the cluster
graph hierarchy for all locally performed iterations $i\geq i_0$, which
defines the tree $\cG_{i_0}$ on clusters $\cC_{i_0}$ (this is feasible as each
$\cG_i$, $i>0$, is also a cluster graph on network graph $\cG_{i-1}$).

This completes the description of the algorithm. Summing up the running times of
the individual steps and using that $i_0=o(\log n)$, we conclude that the
construction takes $(\sqrt{n}+D)n^{o(1)}$ rounds. The construction also
maintained the stated structural properties of the cluster hierarchy. Hence, it
remains to show that (i) we sampled from a distribution of $n^{1+o(1)}$ trees
and (ii) the stated cut approximation properties are satisfied.

Showing these properties now is straightforward. In each step $i>0$ of the
recursion, by \lemmaref{lemma:Hdistribution} we constructed a distribution on
$\tilde{O}(\beta)$ $|\cV_{i-1}|$-trees. The total number of recursive steps
(including the local ones), is bounded by $\lceil \log_{\beta} n\rceil =
O(\log^{1/4} n)$, as $|\cV_i|\leq |\cV_{i-1}|/\beta$ for each $i>0$. On each
level of recursion, we compute a distribution on $2^{O(\sqrt{\log
|\cV_i| \log\log |\cV_i|})}\beta \leq 2^{O(\sqrt{\log n \log \log n})}\beta$
graphs. Hence, the total number of virtual trees in the (implicit) distribution
of virtual trees from which we sampled is bounded by
\begin{equation*}
\left(2^{O(\sqrt{\log n \log \log n})}\beta\right)^{\lceil \log_{\beta} n
\rceil}=n \cdot 2^{O(\sqrt{\log n \log \log n}\log^{1/4}n)}=n^{1+o(1)}.
\end{equation*}
Consider a cut of $G$ of capacity $C$. By the properties of decompositions and
the fact that we multiplied capacities by $1/(1-\varepsilon)$ whenever we
sparsified, $G$ is $1$-embeddable into any of the trees we might construct,
implying that the corresponding cut of the sampled tree has capacity at least
$C$. As in each step, we (i) apply a $(1+\varepsilon)$-sparsifier and multiply
capacities by $1/(1-\varepsilon)$ for constant $\varepsilon$, (ii) construct a
$(2^{O(\sqrt{\log n \log \log n})},\mathbb{H})$-decomposition (for some family
$\mathbb{H}$) from which we sample, and (iii) transform the resulting graph into
a $j$-tree which can be $O(1)$-embedded into the graph from which it is
constructed, we overestimate the capacity of a given cut by an expected factor
of $2^{O(\sqrt{\log n \log \log n})}\cdot O(1)= 2^{O(\sqrt{\log n \log \log
n})}$ in each step. Using that this bound is uniform and the randomness on each
level of recursion is independent, it follows that the expected capacity of a
cut of $G$ of capacity $C$ in the sampled virtual tree is bounded by
\begin{equation*}
\left(2^{O(\sqrt{\log n \log \log n})}\right)^{\lceil \log_{\beta} n
\rceil}=2^{O(\sqrt{\log n \log \log n}\log^{1/4}n)} = n^{o(1)}.
\end{equation*}
\end{proof}



%% file: gradient.tex
\section{The High-Level Algorithm}\label{sec:gradient}
The algorithm is  a distributed implementation of Sherman's algorithm
\cite{Sherman13}. It consists of a logarithmic
number of calls to algorithm $\mathsf{AlmostRoute}$, described
in \sectionref{subsec:almostRoute}, and one computation of a 
maximum-weight
spanning tree and routing the left-over demand through this tree.  
Pseudocode for the top-level algorithm is
presented in \algorithmref{alg:Flow}.
\begin{algorithm}[H]
\caption{Max Flow. Input: demand vector $\vb\in\setR^n$; output: flow 
vector $\vf\in\setR^m$.}
\label{alg:Flow}
\begin{algorithmic}[1]
\small
\State $\vb_0\gets \vb$; $\vf_0\gets\mathbf{0}$
\For{$i\gets1$ to $(\log m+1)$}
	\State $\vf_i \gets \mathsf{AlmostRoute}(\vb_i,\frac12)$
	\State $\vb_i \gets \vb_{i-1}-B\vf_{i-1}$.
\EndFor
\State\label{st-mst} Compute a maximum weight spanning tree $T$ on $G$, 
where weights are the capacities of edges.
\State\label{st-rt-mst} Route the residual demand $\vb_{t}$ through 
$T$; let $\vf_T$ be the resulting flow.
\State Output $\vf_T+\sum_{i=1}^{1+\log m}\vf_i$.
\end{algorithmic}
\end{algorithm}
Most of this section is dedicated to explaining how to implement the 
$\mathsf{AlmostRoute}$
algorithm. Let us  first quickly outline how we implement the final 
steps using standard techniques.
\begin{lemma}\label{lem-subtrees}
	Steps \ref{st-mst}--\ref{st-rt-mst} Can be implemented in the 
	\Congest model in $\tilde O(D+\sqrt n)$ rounds w.h.p.
\end{lemma}
\begin{proof}[\textbf{Proof sketch.}]
A maximum weight spanning tree $T$ can be computed in $\tilde{O}(D+\sqrt{n})$
rounds using the minimum weight spanning tree algorithm of Kutten and
Peleg~\cite{Kutten-Peleg} (say, by assigning weight $w(e):= -\Cp(e)$ for each
edge $e$). To compute the flow, we use the following observation: if $T$ was
rooted at one of its nodes, then to route the demand $\vb_t$ over $T$, it would
be sufficient for each node $v$ to learn the total demand $d_v$ in the subtree
rooted at $v$. In this case each node $v$ assigns $d_v$ units of flow to the
edge leading from $v$ to its parent.
	
We now show how to root the tree and find the total demand in each subtree in
$\tilde{O}(D+\sqrt{n})$ rounds. The algorithm is as follows. Remove each edge of
the tree independently with probability $1/{\sqrt{n}}$. W.h.p.,
\begin{compactenum}[(i)]
\item each connected component induced by the remaining edges contains has
strong diameter $\tilde{O}(\sqrt{n})$,
\item $O(\sqrt{n})$ edges are removed, and hence
\item the number of components is $O(\sqrt{n})$.
\end{compactenum}
Within each component, all demands are summed up, and this sum is made known to
all nodes. The summation takes $\tilde{O}(\sqrt{n})$ rounds due to (i), 
and we
can pipeline the announcement of the sums over a BFS tree in
$\tilde{O}(\sqrt{n}+D)$ rounds due to (iii).

Moreover, in this time we can also assign unique identifiers to the components
(e.g.\ the minimum identifier) and make the tree resulting from contracting
components globally known. Using local computation only, nodes then can root
this tree (e.g.\ at the cluster of minimum identifier) and determine the sum the
demands of the clusters that are fully contained in their subtree. Using a
simple broadcast, the orientation of edges within components is determined, and
using a convergecast on the components, each node can determine the sum of
demands in its subtree. These steps take another $\tilde{O}(\sqrt{n})$ rounds.
\end{proof}

\subsection{Algorithm \texorpdfstring{$\mathsf{AlmostRoute}$}{AlmostRoute}: The
Gradient Descent}
\label{subsec:almostRoute}
We now explain how to implement Algorithm $\mathsf{AlmostRoute}$  in a 
distributed setting. The idea is to use gradient descent with
the potential
function 
$$\phi(\vf)= 
{\Smax(C^{-1}\vf)} 
+ 
{\Smax(2\alpha R(\vb-B\vf))} 
~,$$
where the ``soft-max'' function, defined by
$$\Smax(\boldsymbol y)=
\log\left(\sum_{i=1}^{k} e^{y_i}+e^{-y_i}\right)
~~\text{ for all }\boldsymbol{y}\in\Reals^k~,
$$
is used  as a differentiable approximation to the  ``max'' function. 

 Given this potential function, $\mathsf{AlmostRoute}$ performs
 $O(\alpha^2\eps^{-3}\log n)$ updates on $\vf$ and outputs a flow 
 $\vf$
 that optimizes the potential function up to a $(1+\eps)$ factor.%
 \footnote{Sherman claims that one can save a factor of $1/\eps$ by a more careful scaling~\cite{Sherman13}.}
 Pseduocode for this algorithm is given in \algorithmref{alg:AlmostRoute}.

\begin{algorithm}[H]
\caption{$\mathsf{AlmostRoute(\vb, \eps)}$}
\label{alg:AlmostRoute}
\begin{algorithmic}[1]
\small
\State $k_b\gets 2\alpha\maxnorm{R\vb}\eps/(16\log n)$; $\vb\gets k_b \vb$.
\Repeat
  	\State $k_f\gets 1$
		\While{$\phi(\vf)<16\eps^{-1}\log n$} \label{st-while}
			\State $\vf\gets \vf\cdot(17/16)$; $\vb\gets 
			\vb\cdot(17/16)$; $k_f\gets k_f\cdot(17/16)$ 
			\label{st-17/16}
		\EndWhile
		\State $\delta\gets\sum_{e\in 
		E}|\Cp(e)\frac{\partial\phi}{\partial f_e}|$ \label{st-delta}
		\If{$\delta\ge \eps/4$}
			 \State $f_e\gets 
			 f_e-\sgn\left(\frac{\partial\phi}{\partial f_e}\right)
			 \cdot\Cp(e)\frac{\delta}{1+4\alpha^2}$ \label{st-fe}
		\Else
		  \State $f_e\gets f_e/k_f$ for all edges $e\in E$.
			\State $b_v\gets b_v/(k_bk_f)$ for all nodes $v\in V$
			\State \Return
		\EndIf
\Until{done}
\end{algorithmic}
\end{algorithm}
To implement this algorithm in a distributed setting, we need to compute $R$,
and to do multiplications by $R$ or its transpose $R^{\top}$, distributively.
These multiplications are needed for computing $\phi(\vf)$ and and its partial
derivatives. We remark that $R$ and $R^\top$ are not constructed 
explicitly, as we
need to ensure a small time complexity for each iteration. Assuming that we can
perform these operations, each step of $\mathsf{AlmostRoute}$ can be completed
in $\tilde{O}(D)$ additional rounds.

We maintain the invariant that at the beginning of each iteration of the
\textbf{repeat} loop, each node $v$ knows the current flow over each of the
links $v$ is  incident to, and the current demand at $v$ (i.e., $(\vb-B\vf)_v$).
Let us break the
potential function $\phi$ in two, i.e.,
$$
\phi(\vf)=\phi_1(\vf)+\phi_2(\vf)~\text{,~~~~~where }
\phi_1(\vf)=\Smax(C^{-1}\vf)~~~\text{and}~~~
\phi_2(\vf)={\Smax(2\alpha R(\vb-B\vf))}~.
$$
We proceed as follows. First, we compute
$\phi_1(\vf)$: to find $\Smax(C^{-1}\vf)$, it suffices to
sum $\exp(f_e/\Cp(e))$ and $\exp(-f_e/\Cp(e))$ over all edges $e$, which can be
done in $O(D)$ rounds. As Sherman points out, 
$\phi(\vf) = \Theta(\eps^{-1} \log n)$ due to the scaling, and thus, 
encoding
$\exp(\phi(\vf))$ with sufficient accuracy requires $O(\eps^{-1} \log 
n)$ bits,
which is thereby also a bound on the encoding length of all individual terms in
the sums for $\phi_1$ and $\phi_2$. The error introduced by rounding theses
values to integers is small enough to not affect the asymptotics of the running
time.

For determining $\phi_2(\vf)$, we compute the vector $\vy := 2\alpha 
R(\vb - B\vf)$ and then do an aggregation on a BFS tree as for $\phi_1(\vf)$. 
Since $B\vf$ can be computed instantly ($(B\vf_v$ is exactly the net 
flow into $v$), this boils down to multiplying a locally known vector
with $R$. Before we discuss how implement this operation, let us explain more
about the structure of $R$ and how we  determine
$\frac{\partial \phi}{\partial f_e}$, which is required in 
Lines~\ref{st-delta}
and~\ref{st-fe} of the algorithm.

The linear operator $R$ is induced by graph cuts. More precisely, in
the matrix representation of $R$, there is one row for each cut our congestion
approximator (explicitly) considers. We will clarify the structure of $R$
shortly; for now, denote by $I$ the set of row indices of $R$. Observe that
\begin{equation}
\frac{\partial\phi}{\partial
f_e} =\frac{\exp(f_e/\Cp(e))-\exp(-f_e/\Cp(e))} {\Cp(e) \exp(\phi_1)}
		+\frac{\partial \phi_2}{\partial f_e}\label{eq:del_phi}
\end{equation}
and hence, given that $\phi_1$ is known, the first term is locally computable.
The second term expands to
\begin{align*}
\frac{\partial \phi_2}{\partial f_e} = \sum_{i\in I} \frac{\partial
\phi_2}{\partial y_i} \cdot \frac{\partial y_i}{\partial f_e} = \sum_{i\in I}
\frac{\exp(y_i) - \exp(-y_i)}{\exp(\phi_2)} \cdot \frac{2\alpha B_{i,e}}{\Cp(i)},
\end{align*}
where $\Cp(i)$ is the capacity of cut $i$ in the congestion approximator and
$B_{i,e} \in \{-1,0,1\}$ denotes whether $e$ is outgoing ($-1$), ingoing ($1$),
or not crossing cut $i$.%
\footnote{Technically, $B_{i,e} = \sum_{v\in S_i} B_{ve}$ where $S_i$ is the set of nodes defining cut $i$. 
}

The cuts $i\in I$ are induced by the edges of a collection of (rooted,
virtual, capacitated) spanning trees $\TT$, where for $\cT\in \TT$ we write
$(v,\hat{v})\in \cT$ if $\hat{v}$ is the parent of $v$ and denote by $\cT_v$ the
subtree rooted at $v$. For each $\cT\in \TT$, each edge $(v,\hat{v})\in \cT$ now
induces a (directed) cut $(T_v;\overline{T_v})$ with index 
$i(\cT,(v,\hat{v}))$. We denote the set of edges crossing this cut by 
by $\delta(\cT_v)$. Let us also define
$$
p(\cT,v)=\frac{\exp(y_{i(\cT,(v,\hat{v}))}) -
	\exp(-y_{i(\cT,(v,\hat{v}))})}{\exp(\phi_2)}
\cdot \frac{2\alpha}{\Cp_{\cT}((v,\hat{v}))}~.
$$
With this notation, we have that
\[
\frac{\partial \phi_2}{\partial f_e} 
= \sum_{\cT \in \TT} \sum_{\substack{(v,\hat{v})\in \cT\\
e \in \delta(\cT_v)}} 
p(\cT,v)\cdot
  B_{i(\cT,(v,\hat{v})),e}.
\]
We call $p(\cT,(v,\hat{v}))$ the \emph{price} of the (virtual) edge
$(v,\hat{v})\in \cT$. Let $\cP_{v,\cT}$ denote the edge set of the unique path
in $\cT$ from $v$ to the root of $\cT$. We define a \emph{node 
potential} for each node $v$ by
\[
 \pi_v := \sum_{\cT \in \TT} \sum_{(w,\hat{w}) \in \cP_{v,\cT}}
 p(\cT,(w,\hat{w}))~.
\]
For any $e=(u,v)$, the cuts induced by edges in $\cT\in \TT$
that $e$ crosses correspond to the edges on the unique path from $u$ to $v$ in
$\cT$. For all edges $(w,\hat{w})\in \cT$ on the path from $u$ to the least
common ancestor of $u$ and $v$ in $\cT$, $B_{i(\cT,(w,\hat{w})),e}=-1$, while
$B_{i(\cT,(w,\hat{w})),e}=+1$ for the edges on the path between $v$ and this
least common ancestor. Thus,
\begin{equation}
\frac{\partial \phi_2}{\partial f_e} 
= \pi_v - \pi_u\,,\label{eq:del_phi_2_potential}
\end{equation}
and our task boils down to determining the value of the potential $\pi_v$ at
each node $v\in V$.
To this end, we need two key subroutines to compute distributively the 
following quantities.  
\begin{enumerate}[(1)]
\item $y_i$ for each cut $i$. Note that $\vb-B\vf$ is known distributedly, i.e., each 
node knows its own coordinate of this vector. For each tree in $\cT\in \TT$, we
need to aggregate this information from the leaves to the root. This means to
simulate a convergecast on the virtual tree $\cT$.
\item $\pi_v$ for each node $v$. Provided that each (virtual) tree edge knows
its $y$-value and $\phi_2$, the prices can be computed locally. Then the
contribution of each tree to the node potentials can be computed by a downcast
from the corresponding root to its leaves.
\end{enumerate}
With these routines, one iteration of the \textbf{repeat} loop is now executed
as follows:
\begin{compactenum}
\item Compute $\phi_1$, $\vy$ (local knowledge), and $\phi_2$ (aggregation on
BFS tree once $\vy$ is known).
\item Check the condition in Line~\ref{st-while}. If it holds, locally
update $\vb$, $\vf$, and $k_f$, and go to the previous step.
\item Compute the potential $\boldsymbol{\pi}$ (local knowledge).
\item For each $e\in E$, its incident edges determine $\frac{\partial
\phi}{\partial f_e}$ (based on Equations~\ref{eq:del_phi}
and~\ref{eq:del_phi_2_potential}, it suffices to exchange $\pi_u$ and
$\pi_v$ over $e$).
\item Compute $\delta$ (aggregation on BFS tree).
\item Locally update $f_e$ and $b_v$ for all $e\in E$ and $v\in V$.
\end{compactenum}
Note that all of the individual operations except for computation of $\vy$ and
$\boldsymbol{\pi}$ can be completed in $O(D)$ rounds. 
Sherman proved~\cite{Sherman13} that $\mathsf{AlmostRoute}$ terminates after 
$\tilde{O}(\eps^{-3}\alpha^2)$ iterations. As it is only called $O( \log n )$
times by the max-flow algorithm, 
\theoremref{theorem:main_informal} follows if we can compute $\vy$ and
$\boldsymbol{\pi}$ in $(\sqrt{n}+D)n^{o(1)}$ rounds for an $\alpha$-congestion
approximator with $\alpha=n^{o(1)}$; this is subject of the next subsection.

\subsection{Congestion Approximation}
\label{subsec:congestionApprox}
\begin{figure}[t]
	\centering
		\includegraphics[width=0.80\textwidth]{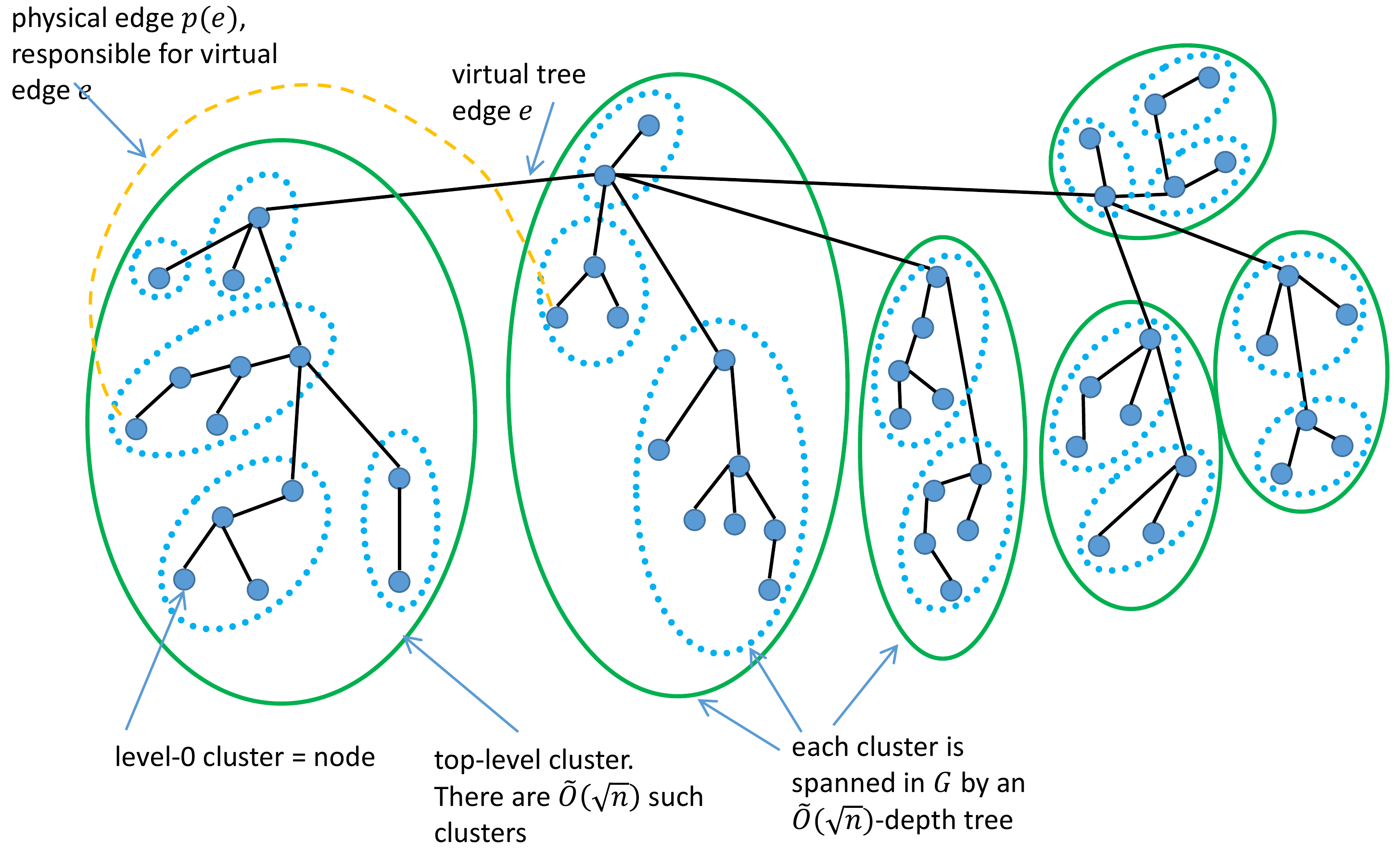}
	\caption{\small Hiearchical cluster decomposition of a virtual tree
	$\cT\in \TT$. Black edges are virtual tree edges, which are represented by a
	physical edge connecting the top-level clusters they connect (the orange
	dotted edge $p(e)$ corresponds to the edge labeled $e$). Each cluster is
	spanned by a tree in $G$ of depth $\tilde{O}(\sqrt{n})$, which is not shown.}
	\label{fig:hierarchical}
\end{figure}

Our congestion approximator $R$ is defined by the edge-induced cuts of a
sample $\TT$ of virtual trees $\cT$ from a recursively constructed distribution.
The trees are represented distributedly by a hierarchy of \emph{cluster graphs}
(see \figureref{fig:hierarchical} for an illustration and
\sectionref{sec:cluster} for the formal definition of cluster graphs).
Intuitively, a cluster graph partitions the nodes into clusters, each of which
has a spanning tree rooted at a leader, and a collection of edges between
clusters that are represented by corresponding graph edges between some nodes of
the clusters they connect. In \sectionref{sec:jtrees}, we have shown the
following theorem.
\begin{theoremR}{theorem:distribution}
W.h.p., within $(\sqrt{n}+D)n^{o(1)}$ rounds of the \Congest model, we
can sample a tree $\cT$ from a distribution of $n^{1+o(1)}$ (virtual)
rooted spanning trees on $G$ with the following properties.
\begin{compactitem}
\item For any cut of $G$ of capacity $C$, the capacity of the cut in $\cT$ is
at least $C$.
\item For any cut of $G$ of capacity $C$, the expected capacity of the cut in
$\cT$ is at most $\alpha C$, where $\alpha\in n^{o(1)}$.
\item The distributed representation of $\cT$ is given by a hierarchy of cluster
graphs $\cG_i=(\cV_i,\cE_i,\cL_i,\mathfrak{T}_i,\psi_i)$,
\ $i\in \{0,\ldots,i_0\}$, $i_0\in o(\log
n)$, on network graph $G$, with the following properties.
\begin{compactitem}
\item The spanning trees of the clusters of $\cG_i$ have depth
$\tilde{O}(\sqrt{n})$.
\item $|\cV_{i_0}|=n^{1/2+o(1)}$.
\item $\cG_i$ is the (rooted) tree resulting from $\cT$ by contracting the
clusters of $\cG_i$.
\item For $i>0$, $\cG_i$ is also a cluster graph on network graph $\cG_{i-1}$.
\item For $i>0$, each cluster $c_i\in \cV_i$ of $\cG_i$, interpreted as cluster
graph on $\cG_{i-1}$, contains a unique \emph{portal cluster} $p(c_i)\in
\cV_{i-1}$ of $\cG_{i-1}$ that is incident\footnote{Note that the corresponding
physical edges in $G$ may still connect to different sub-clusters of $c_i$.} to
all edges of $\cG_i$ containing $c_i$. That is, $\cG_{i-1}$ is a $|\cV_i|$-tree
with core $p(\cV_i)$.
\end{compactitem}
\end{compactitem}
\end{theoremR}
The first two properties of each $\cT$ stated in the theorem imply that we can
use them to construct a good congestion approximator $R$. More precisely,
\lemmaref{lemma:approximator} implies the following corollary.
\begin{corollary}\label{cor:iterations}
Sampling a collection $\TT$ of $O(\log n)$ virtual trees given by
\theoremref{theorem:sample} and using them as congestion approximator $R$ in the
way specified in \sectionref{subsec:almostRoute} implies that the total number
of iterations of \algorithmref{alg:AlmostRoute} is $n^{o(1)}$.
\end{corollary}

All that remains now is to show that the distributed representation of each
sampled $\cT\in \TT$ allows to simulate a convergecast and a downcast on $\cT$
in $(\sqrt{n}+D)n^{o(1)}$ rounds: then we can implement the key subroutines (1)
and (2) (i.e., compute $\vy$ and $\boldsymbol{\pi}$) outlined in
\sectionref{subsec:almostRoute} with this time complexity, and by
\corollaryref{cor:iterations} the total number of rounds of the computation is
bounded by $(\sqrt{n}+D)n^{o(1)}$.

Fortunately, the recursive structure of the decomposition is very specific. The
cluster graphs of the different levels of recursion are nested, i.e., the
clusters of the $(i-1)^{th}$ level of recursion are subdivisions of the clusters
of the $i^{th}$ level. What is more, each cluster is a subtree of the virtual
tree and is spanned by a tree of depth $\tilde{O}(\sqrt{n})$ in $G$ (cf.\
\figureref{fig:hierarchical}) Hence, while the physical graph edges representing
the virtual tree edges are between arbitrary nodes within the clusters they
connect, we can (i) identify each cluster on each hierarchy level with the root
of the subtree induced by its nodes, (ii) handle such subtrees recursively (both
for convergecasts and downcasts), (iii) on each level of recursion but the last,
perform the relevant communication by broadcasting or upcasting on the
underlying cluster spanning trees in $G$ of depth $\tilde{O}(\sqrt{n})$, and
(iv) communicate over a BFS tree of $G$ on the final level of recursion, where
merely $n^{1/2+o(1)}$ clusters/nodes of the virtual tree remain.
\begin{corollary}\label{cor:virtual_upcast_convergecast}
On each virtual tree $\cT\in \TT$, we can simulate convergecast and upcast
operations in $\tilde{O}(\sqrt{n}+D)$ rounds.
\end{corollary}
\theoremref{theorem:main_informal} now follows from Sherman's results on the
number of iterations of the gradient descent algorithm~\cite{Sherman13}, the
discussion in \sectionref{subsec:almostRoute}, and
Corollaries~\ref{cor:iterations} and~\ref{cor:virtual_upcast_convergecast}.